\documentclass[10pt,twocolumn,twoside]{IEEEtran}
\usepackage{cite}  
\usepackage{bm} 
\usepackage[pdftex]{graphicx}
\usepackage{stfloats} 
\usepackage{amsmath,amsfonts,amssymb,amscd,empheq}
\usepackage{hyperref}

\usepackage{multirow}
\usepackage{bbm}
\usepackage{bbold}

\usepackage{amsthm}
\usepackage{upgreek}

\newtheorem{proposition}{Proposition}

\newtheorem{lemma}{Lemma}
\newtheorem{remark}{Remark}

\newtheorem{assumption}{Assumption}

\usepackage{setspace}
\usepackage{pdfsync}
\usepackage{epstopdf}
\usepackage{lipsum,multicol}
\usepackage[numbered,autolinebreaks,useliterate]{mcode}
\usepackage{mathtools,bbm}
\usepackage{mathtools, dsfont}

\usepackage{caption}
\usepackage{subcaption}
\usepackage{scalefnt}

\usepackage{lipsum}
\usepackage{calligra}
\DeclareMathAlphabet{\mathcalligra}{T1}{calligra}{m}{n}

\DeclareMathOperator{\diag}{diag}

\newcommand{\vertiii}[1]{{\left\vert\kern-0.25ex\left\vert\kern-0.25ex\left\vert #1 
    \right\vert\kern-0.25ex\right\vert\kern-0.25ex\right\vert}}

\usepackage[noend]{algpseudocode}
\usepackage{algorithm,algorithmicx}

\algrenewcommand\algorithmicrequire{\textbf{Initialization:}}
\algrenewcommand\algorithmicensure{\textbf{For each $i$, execute for $k\geq 0$:}}

\algblockdefx[Loop]{Loop}{EndLoop}[1]{\textbf{for $k\geq 0$:} #1}

\makeatletter
\algrenewcommand\ALG@beginalgorithmic{\small}
\makeatother

\makeatletter
\newcommand\fs@betterruled{%
  \def\@fs@cfont{\bfseries}\let\@fs@capt\floatc@ruled
  \def\@fs@pre{\vspace*{5pt}\hrule height.8pt depth0pt \kern2pt}%
  \def\@fs@post{\kern2pt\hrule\relax}%
  \def\@fs@mid{\kern2pt\hrule\kern2pt}%
  \let\@fs@iftopcapt\iftrue}
\floatstyle{betterruled}
\restylefloat{algorithm}
\makeatother

\usepackage{cases}
\usepackage{balance}
\usepackage{float}
\usepackage{tabu}

\allowdisplaybreaks
\begin{document}
\title{{Fast Distributed Coordination of Distributed Energy Resources over Time-Varying Communication Networks}}

\author{Madi Zholbaryssov, Christoforos N. Hadjicostis,~\IEEEmembership{Fellow,~IEEE}, \\Alejandro D. Dom\'{i}nguez-Garc\'{i}a,~\IEEEmembership{Senior Member,~IEEE}
        \thanks{M. Zholbaryssov and A. D. Dom\'{i}nguez-Garc\'{i}a are with the ECE Department at the University of Illinois at Urbana-Champaign, Urbana, IL 61801, USA. E-mail:  \{zholbar1, aledan\}@ILLINOIS.EDU.}
\thanks{C. N. Hadjicostis is with the ECE Department at the University of Cyprus, Nicosia, Cyprus, and also with the ECE Department at  the University of Illinois at Urbana-Champaign, Urbana, IL 61801, USA. E-mail:  chadjic@UCY.AC.CY.}
}
\maketitle
\begin{abstract}
In this paper, we consider the problem of optimally coordinating the response of a group of distributed energy resources (DERs) so they collectively meet the electric power demanded by a collection of loads, while minimizing the total generation cost and respecting the DER capacity limits. This problem can be cast as a convex optimization problem, where the global objective is to minimize a sum of convex functions corresponding to individual DER generation cost, while satisfying (i) linear inequality constraints corresponding to the DER capacity limits and (ii) a linear equality constraint corresponding to the total power generated by the DERs being equal to the total power demand. We develop distributed algorithms to solve the DER coordination problem over time-varying communication networks with either bidirectional or unidirectional communication links. The proposed algorithms can be seen as distributed versions of a centralized primal-dual algorithm. One of the algorithms proposed for directed communication graphs has geometric convergence rate even when communication out-degrees are unknown to agents.
We showcase the proposed algorithms using the standard IEEE~$39$--bus test system, and compare their performance against other ones proposed in the literature.
\end{abstract}
\IEEEpeerreviewmaketitle

\section{Introduction}
It is envisioned that present-day power grids, which are dependent on centralized power generation stations, will transition towards more decentralized power generation mostly based on distributed energy resources (DERs). One of the obstacles in making this shift happen is to find effective control strategies for coordinating DERs. In this regard, and partly due to high variability introduced by renewable-based generation resources, DERs will need to more frequently adjust their set-points, which entails development of fast control strategies. Also, because of the communication overhead, it may not be feasible to use a centralized approach to coordinate a large number of DERs over a large geographic area. This necessitates DER coordination using distributed control strategies that scale well to power networks of large size.

In this work, we consider a group of DERs and electrical loads, which are interconnected by an electric power network, and can exchange information among themselves via some communication network. Each DER is endowed with a power generation cost function, which is unknown to other DERs, and its power output is upper- and lower-limited by some capacity constraints. A computing device attached to each DER is able to communicate with the computing devices of other DERs located within its communication range. Then, the objective is to determine, in a distributed manner, DER optimal power outputs so as to satisfy total electric power demand while minimizing the total generation cost and respecting DER capacity limits. This DER coordination problem can be cast as a convex optimization problem (see, e.g., \cite{ZhCh12, DoCaHa12,KaHu12,ZhPa14,CaDoHa15,ChZh16,WuJo17}), where the global objective is to minimize a sum of convex functions corresponding to the costs of generating power from the DERs, while satisfying linear inequality constraints on the power produced by each DER, and a linear equality constraint corresponding to the total generated power being equal to the total power consumed by the electrical loads. 

Since we aim to solve the DER coordination problem in a distributed manner, we also address the issue of achieving resilient and fault-tolerant operation, which requires a control design that is robust to communication delays and random data packet losses. In this paper, we focus on the challenges that arise due to the time-varying nature of the underlying communication network, and address the DER coordination problem via distributed algorithms that are capable of operating over time-varying communication graphs with either (i) bidirectional or (ii) unidirectional communication links. These algorithms also have geometric convergence rate, which is a desirable feature for ensuring fast performance.
We believe that the proposed algorithms can be extended to solve more complex DER coordination problems with additional constraints, e.g., line flow constraints, voltage constraints, or reactive power balance constraints, as long as these are linear and have a separable structure, i.e., each constraint is local or involves only a pair of neighboring nodes.

A vast body of work has focused on solving the DER coordination problem in a distributed way (see, e.g., \cite{ZhCh12, DoCaHa12,YaTa13,KaHu12, ZhPa14, CaDoHa15,ChCo15,ChZh16,WuJo17,DuYa18,YaWu19}). Earlier works focused on time-invariant communication networks (see, e.g., \cite{ZhCh12, DoCaHa12,YaTa13,ChCo15}). In one of the earliest works, the authors of \cite{ZhCh12} proposed a distributed approach in which agents' local estimates are driven to the optimal incremental cost via the leader-follower consensus algorithm. The authors of \cite{DoCaHa12} utilize the so-called \textit{ratio-consensus algorithm} (see, e.g., \cite{DoHa12,HaVaDo16}) to distributively compute the solution to the dual formulation of the DER coordination problem. Later works focused on time-varying communication networks (see, e.g., \cite{KaHu12, ZhPa14, CaDoHa15,ChZh16,WuJo17}). For example, in \cite{WuJo17}, the authors propose a robustified version of the so-called subgradient-push method (see, e.g., \cite{Nedic15}) that operates over time-varying directed communication networks; the algorithm utilizes the so-called push-sum protocol (see, e.g., \cite{Kempe03,BeBl10,DoHa11}) to converge to a consensual solution. In \cite{KaHu12}, the authors propose a distributed algorithm that uses a consensus term to converge to a common incremental cost, and a subgradient term to satisfy the total load demand; the algorithm is designed assuming that generation cost functions are quadratic. However, convergence of the algorithms proposed in \cite{WuJo17} and \cite{KaHu12} is not guaranteed to be geometrically fast and might be slow due to the fact that the algorithms use a diminishing stepsize. In \cite{DuYa18} and \cite{YaWu19}, the authors propose distributed algorithms based on the dual-ascent method that have geometric convergence rate but require the agents to know their communication out-degrees.

Our starting point in the design of the algorithms is a primal-dual algorithm (first order Lagrangian method), where the dual variable associated with the power balance constraint depends on the total power imbalance (supply-demand mismatch). We then develop distributed versions of this primal-dual algorithm by having DERs closely emulate the iterations of the primal-dual algorithm. To this end, each node with a DER maintains an estimate of the dual variable and updates it using a local estimate of the total power imbalance and the neighbors' estimates of the total power imbalance. The update of the total power imbalance estimate is based on the gradient tracking idea that appeared in \cite{Nedic17}.
To enable agents to operate over time-varying directed communication graphs when their communication out-degrees are unknown to them, we propose a robust distributed primal-dual algorithm that converges geometrically fast. 

Each proposed algorithm is viewed as a feedback interconnection of the (centralized) primal-dual algorithm representing the nominal system and the error dynamics due to the nature of the distributed implementation. The key ingredient for establishing the convergence results is to show that both systems are finite-gain stable, which then allows us to use the small-gain theorem (see, e.g., \cite{Khalil}) to show the convergence of the feedback interconnected system. The small-gain-theorem-based analysis first appeared in \cite{Nedic17} in the context of distributed algorithms for solving an unconstrained consensus optimization problem.
\section{Preliminaries}
\label{sec:preliminaries}
In this section, we formulate the DER coordination problem and give an overview of the small-gain theorem for discrete-time systems.
\subsection{DER Coordination Problem}
We consider a collection of DERs and electrical loads interconnected by a power network. Let $p_i$ denote the power output of the DER at bus $i$, $1\leq i\leq n$, and let $\ell_i$ denote the power consumed by the load at bus $i$, $1\leq i\leq n$. Let $\underline{p}_i$ and $\overline{p}_i$ denote the lower and upper limits on the power that the DER at bus $i$ can generate. Also, let $f_i(\cdot)$ denote the cost function associated with the electric power generated by the DER at bus $i$. We assume the power consumed by the loads is fixed and known. Then, our main objective is to determine power generated by the DERs in order to collectively satisfy total electric power demand, $\sum_{i=1}^n \ell_i$, while minimizing the total generation cost, $\sum_{i=1}^n f_i(p_i)$. 

 
More formally, we consider the following DER coordination problem that has been studied in \cite{ZhCh12, DoCaHa12,KaHu12,ZhPa14,ChZh16,WuJo17}:
\begin{subequations}\label{ED}
\begin{align}
\underset{p\in\mathds{R}^n}{\mbox{minimize}} & \mbox{ } \sum\limits_{i=1}^n f_i(p_i) \label{cost}\\
\mbox{subject to} & \mbox{ } \mathbf{1}^\top p = \mathbf{1}^\top \ell, \label{flow_constr}\\
&\mbox{ }\underline{p}\leq p\leq \overline{p},\label{box_constr1}
\end{align}
\end{subequations}
where $p=[p_1,\dots,p_n]^\top$, $\ell=[\ell_1,\dots,\ell_n]^\top$, $\underline{p} = [\underline{p}_1,\dots,\underline{p}_n]^\top$, $\overline{p} = [\overline{p}_1,\dots,\overline{p}_n]^\top$, and $\mathbf{1}$ is the all-ones vector (its size should be clear from the context). We assume that $\mathbf{1}^\top\underline{p} \leq \mathbf{1}^\top\ell \leq \mathbf{1}^\top\overline{p}$, which makes~\eqref{ED} feasible. 
Additionally, we make the following assumption regarding the objective function. 
\begin{assumption}\label{objective_assumption}
Each cost function $f_i(\cdot)$ is twice differentiable and strongly convex with parameter $m>0$, i.e., $f_i^{\prime\prime}(x)\geq m$, $\forall x \in [\underline{p}_i, \overline{p}_i]$, $\forall i\in\mathcal{V}$.
\end{assumption}
The main objective of our work in this paper is to design a distributed algorithm for solving~\eqref{ED} geometrically fast over time-varying communication networks.
\subsection{The Small-Gain Theorem}
In the following, we give a brief overview of the main analysis tool used in later developments---the small-gain theorem (see, e.g., \cite[Theorem~5.6]{Khalil}) for discrete-time systems.
For the forthcoming developments, we adopt the appropriate metric for measuring energy content of the signals of interest. For a given sequence of iterates, $\{x[k]\}_{k=0}^{\infty}$, where $x[k]\in\mathds{R}^n$, consider the following norm (previously used in \cite{Nedic17}): \[\|x\|_2^{a,K}\coloneqq \max\limits_{0\leq k\leq K}a^{-k}\|x[k]\|_2,\] for some $a \in (0,1)$, where $\|\cdot\|_2$ is the Euclidean norm. If $\|x\|_2^{a,K}$ is bounded for all $K\geq 0$, then, $a^{-k}\|x[k]\|_2$ is bounded for all $k\geq 0$, and, thus, it follows that $x[k]$ converges to zero at a geometric rate $\mathcal{O}(a^k)$.

Now, consider a feedback connection of two discrete-time systems~$\mathcal{H}_1$ and $\mathcal{H}_2$ such that
\begin{align*}
e_2[k+1] &= \mathcal{H}_1(e_1[k]),\\
e_1[k+1] &= \mathcal{H}_2(e_2[k]).
\end{align*}
We assume that $\mathcal{H}_1$ and $\mathcal{H}_2$ are finite-gain stable in the sense of the norm $\|\cdot\|_2^{a,K}$, namely, the following relations hold:
\begin{subequations}\label{finite-gain-stability}
\begin{align}
\|e_2\|_2^{a,K} &\leq \gamma_1\|e_1\|_2^{a,K}+\beta_1,\\
\|e_1\|_2^{a,K} &\leq \gamma_2\|e_2\|_2^{a,K}+\beta_2,
\end{align}
\end{subequations}
for some nonnegative constants $\beta_1$, $\beta_2$, $\gamma_1$, and $\gamma_2$.
From~\eqref{finite-gain-stability}, we have that
\begin{align}\label{small-gain-th-1}
\|e_2\|_2^{a,K} &\leq \gamma_1\|e_1\|_2^{a,K}+\beta_1\nonumber\\
&\leq \gamma_1\gamma_2\|e_2\|_2^{a,K}+\gamma_1\beta_2+\beta_1,
\end{align}
which by rearranging yields
\begin{align*}
\|e_2\|_2^{a,K} &\leq \frac{\gamma_1\beta_2+\beta_1}{1-\gamma_1\gamma_2}.
\end{align*}
Similarly,
\begin{align*}
\|e_1\|_2^{a,K} &\leq \frac{\gamma_2\beta_1+\beta_2}{1-\gamma_1\gamma_2}.
\end{align*}
Then, if $\gamma_1\gamma_2<1$, $\|e_1\|_2^{a,K}$ and $\|e_2\|_2^{a,K}$ are bounded, and $e_1[k]$ and $e_2[k]$ converge to zero at a geometric rate $\mathcal{O}(a^k)$. 

\section{DER Coordination Over Time-Varying Undirected Graphs}\label{sec:DOD_un}
In this section, we present a distributed algorithm for solving the DER coordination problem~\eqref{ED} over time-varying undirected communication graphs. 
\subsection{Communication Network Model}\label{subsec:cyber_layer_un}
Here, we introduce the model describing the communication network that enables the bidirectional exchange of information between DERs. 
Let $\mathcal{G}^{(0)} = (\mathcal{V},\mathcal{E}^{(0)})$ denote an undirected graph, where each element in the node set $\mathcal{V} \coloneqq \{1,2,\dots,n\}$ corresponds to a DER, and $\{i,j\} \in \mathcal{E}^{(0)}$ if there is a communication link between DERs $i$ and $j$ that allows them to exchange information.
During any time interval $(t_k,t_{k+1})$, successful data transmissions among the DERs can be captured by the undirected graph $\mathcal{G}^{(c)}[k]=(\mathcal{V},\mathcal{E}^{(c)}[k])$, where $\mathcal{E}^{(c)}[k]\subseteq \mathcal{E}^{(0)}$ is the set of active communication links, with $\{i,j\} \in \mathcal{E}^{(c)}[k]$ if nodes $i$ and $j$ simultaneously exchange information with each other during time interval $(t_k,t_{k+1})$. 
Let $\mathcal{N}_i \coloneqq\{j\in\mathcal{V}:\{i,j\} \in\mathcal{E}^{(0)}\}$ denote the set of nominal neighbors, and $d_i \coloneqq |\mathcal{N}_i|+1$ the nominal degree of node~$i$.
Let $\mathcal{N}_i[k]$ denote the set of neighbors of node~$i$ during time interval $(t_k,t_{k+1})$, i.e., $\mathcal{N}_i[k]\coloneqq\{j\in\mathcal{V}:\{i,j\} \in \mathcal{E}^{(c)}[k]\}$.
We make the following standard assumption regarding the connectivity of the network (see, e.g., \cite{Nedic15,Nedic17}). 
\begin{assumption}
\label{assume_comm_model_un1}
There exists some positive integer $B$ such that the graph with node set $\mathcal{V}$ and edge set $\bigcup_{l= kB}^{(k+1)B-1}\mathcal{E}^{(c)}[l]$ is connected for $k=0,1,\dots$. 
\end{assumption}

\subsection{Distributed Primal-Dual Algorithm}
Our starting point to solve~\eqref{ED} is the following primal-dual algorithm \cite[Chapter~4.4]{NonlinearProgramming} with the additional projection:
\begin{subequations}\label{central_alg}
\begin{align}
p_i[k+1]&=\Big[p_i[k]-s f^\prime_i(p_i[k])+s\xi\overline{\lambda}[k]\Big]_{\underline{p}_i}^{\overline{p}_i},\\
\overline{\lambda}[k+1] &= \overline{\lambda}[k]-s\mathbf{1}^\top (p[k]-\ell),
\end{align}
\end{subequations}
where $p[k]=[p_1[k],\dots,p_n[k]]^\top$, $[\cdot]_{\underline{p}_i}^{\overline{p}_i}$ denotes the projection onto the interval $[\underline{p}_i,\overline{p}_i]$, $s>0$ is a constant stepsize, $\xi \in (0,1]$ is a constant parameter, and $\overline{\lambda}[k]$ is the estimate of the Lagrange multiplier at time $k$ associated with the power balance constraint, $\mathbf{1}^\top p = \mathbf{1}^\top\ell$. 
Algorithm~\eqref{central_alg} does not conform to the general communication model described in Section~\ref{subsec:cyber_layer_un} because in order to execute it, the total power imbalance, $\mathbf{1}^\top (p[k]-\ell)$, at time $k$ is needed to update $\overline{\lambda}[k]$. 

To design a distributed version of~\eqref{central_alg}, each node~$i$ needs to have a local estimate of $\overline{\lambda}[k]$, denoted by $\lambda_i[k]$. To update $\lambda_i[k]$, it should also have an estimate of $\mathbf{1}^\top (p[k]-\ell)$. One such estimate that can be constructed purely based on the local power imbalance is $\hat{n}(p_i[k]-\ell_i)$, where $\hat{n}$ is some estimate of $n$ that every node has, e.g., $\hat{n}$ can be one, which leads us to the following distributed algorithm:
\begin{subequations}\label{primal_dual_un_ac}
\begin{align}
p_i[k+1]&=\Big[p_i[k]-sf^\prime_i(p_i[k])+s\xi\lambda_i[k]\Big]_{\underline{p}_i}^{\overline{p}_i},\\
\lambda_i[k+1]&= (1-\sum_j a_{ij}[k])\lambda_i[k]+\sum_j a_{ij}[k]\lambda_j[k]\nonumber\\&\quad-s\hat{n}(p_i[k]-\ell_i),\label{lmbd_update}
\end{align}
\end{subequations}
where $a_{ij}[k]= a_{ji}[k]\geq \eta$ if $\{i,j\} \in \mathcal{E}^{(c)}[k]$, $a_{ij}[k]=0$ if $\{i,j\} \notin \mathcal{E}^{(c)}[k]$, and the constant $\eta > 0$ is chosen so that $1-\sum_j a_{ij}[k]\geq\eta$. Even if $\hat{n}$ is an accurate estimate of $n$, $\hat{n}(p_i[k]-\ell_i)$ is a very crude estimate of $\mathbf{1}^\top (p[k]-\ell)$, and results in poor performance as will be demonstrated later via numerical simulations. 

A better approach is to let each node estimate the total power imbalance by using its local power imbalance and the estimates of its neighbors. 
To elaborate on this further, we let $y_i$ denote node~$i$'s estimate of the total power imbalance. Then, one way to update $y_i$ is as follows:
\begin{align}
y_i[k+1]&= (1-\sum_j a_{ij}[k])y_i[k]+\sum_j a_{ij}[k]y_j[k]\nonumber\\&\quad+\hat{n}(p_i[k+1]-p_i[k]), \label{y_update}
\end{align}
where $y_i[0] = \hat{n}(p_i[0]-\ell_i)$. In~\eqref{y_update}, node~$i$ first computes the average of its estimate and the estimates of its neighbors, and then adds $\hat{n}(p_i[k+1]-p_i[k])$ to ensure that the average of all total power imbalance estimates is always equal to $({\hat{n}}/{n})\mathbf{1}^\top (p[k]-\ell)$, which is equal to the total power imbalance, $\mathbf{1}^\top (p[k]-\ell)$ if $\hat{n}=n$. This second step allows local estimates to remain close to the total power imbalance. 
Below, we provide the complete update formula for the primal and dual variables:
\begin{subequations}
\begin{align}
p_i[k+1]&=\Big[p_i[k]-sf^\prime_i(p_i[k])+s\xi\lambda_i[k]\Big]_{\underline{p}_i}^{\overline{p}_i},\label{u_update_un}\\
\lambda_i[k+1]&= (1-\sum_j a_{ij}[k])\lambda_i[k]+\sum_j a_{ij}[k]\lambda_j[k]\nonumber\\&\quad-sy_i[k],\label{lmbd_update}\\
y_i[k+1]&= (1-\sum_j a_{ij}[k])y_i[k]+\sum_j a_{ij}[k]y_j[k]\nonumber\\&\quad+\hat{n}(p_i[k+1]-p_i[k]).\label{y_update_un}
\end{align}
\label{primal_dual_un}%
\end{subequations}
Note that in~\eqref{lmbd_update}, node~$i$ computes the (weighted) average of its estimate and the estimates of its neighbors, which yields a good estimate of $\overline{\lambda}$. 

\subsection{Feedback Interconnection Representation of the Distributed Primal-Dual Algorithm}
In the following, we represent~\eqref{primal_dual_un} as a feedback interconnection of a nominal system, denoted by $\mathcal{H}_1$, and a disturbance system, denoted by $\mathcal{H}_2$, which allows us to utilize the small-gain theorem for convergence analysis purposes. To this end, let $e[k] \coloneqq \lambda[k]-(\frac{1}{n}\mathbf{1}^\top\lambda[k])\mathbf{1}$, $\lambda[k]=[\lambda_1[k],\dots,\lambda_n[k]]^\top$, and $\hat{\lambda}[k] \coloneqq \frac{1}{\hat{n}}\mathbf{1}^\top\lambda[k]$; then, we  define the nominal system, $\mathcal{H}_1$, as follows:
\begin{subequations}\label{H1}
\begin{numcases}{\mathcal{H}_1:}
\begin{aligned}p[k+1]&=\Big[p[k]-s\nabla f(p[k])+s\xi\frac{\hat{n}}{n}\mathbf{1}\hat{\lambda}[k]\\&\quad+s\xi e[k]\Big]_{\underline{p}}^{\overline{p}},\end{aligned}\label{H1_1}\\
\hat{\lambda}[k+1] = \hat{\lambda}[k]-s\mathbf{1}^\top (p[k]-\ell),\label{H1_2}
\end{numcases}
\end{subequations}
where $p[k] = [p_1[k],p_2[k],\dots,p_n[k]]^\top$, $\nabla f(p[k]) = [f_1^\prime(p_1[k]),f_2^\prime(p_2[k]),\dots,f_n^\prime(p_n[k])]^\top$, $\ell = [\ell_1,\ell_2,\dots,\ell_n]^\top$, and $[\cdot]_{\underline{p}}^{\overline{p}}$ denotes the projection onto the box $[\underline{p},\overline{p}]$. Note that in order to obtain~\eqref{H1_1}, we substituted $e[k]+\frac{\hat{n}}{n}\mathbf{1}\hat{\lambda}[k]$ for $\lambda[k]$ in~\eqref{u_update_un}, and summed~\eqref{lmbd_update} over all $i$ and divided the result by $\hat{n}$ to obtain~\eqref{H1_2}.
We note that $e[k]$ is the vector of deviations of the local estimates of the Lagrange multiplier from their average at time instant $k$; without $e[k]$, the nominal system~$\mathcal{H}_1$ has almost the same form as~\eqref{central_alg}.
Now, we define the disturbance system, $\mathcal{H}_2$, as follows: 
\begin{subequations}\label{H2}
\begin{numcases}{\mathcal{H}_2:}
\begin{aligned}y[k]&=W[k-1]y[k-1]\\&\quad+\hat{n}(p[k]-p[k-1]),\end{aligned}\label{H2_1}\\
\lambda[k+1]= W[k]\lambda[k]-sy[k],\label{H2_2}\\
e[k] = \lambda[k]-\mathbf{1}\frac{\hat{n}}{n}\hat{\lambda}[k],\label{H2_3}
\end{numcases}
\end{subequations}
where $W[k]=[w_{ij}[k]]\in\mathds{R}^{n\times n}$ is a weight matrix at time instant $k$ with $w_{ij}[k]=a_{ij}[k]$, $j\neq i$, and $w_{ii}[k]=1-\sum_j a_{ij}[k]$, $i\in\mathcal{V}$.
Then, as illustrated in Fig.~\ref{fig:control_interpret}, algorithm~\eqref{primal_dual_un} can be viewed as a feedback interconnection of~$\mathcal{H}_1$ and $\mathcal{H}_2$,
where $(p^*,\lambda^*)$ is the equilibrium of~\eqref{H1} when $e[k]\equiv 0$, for all $k\geq0$.
\begin{figure}
    \centering 
	\includegraphics[trim=0cm 0cm 0cm 0cm, clip=true, scale=0.6]{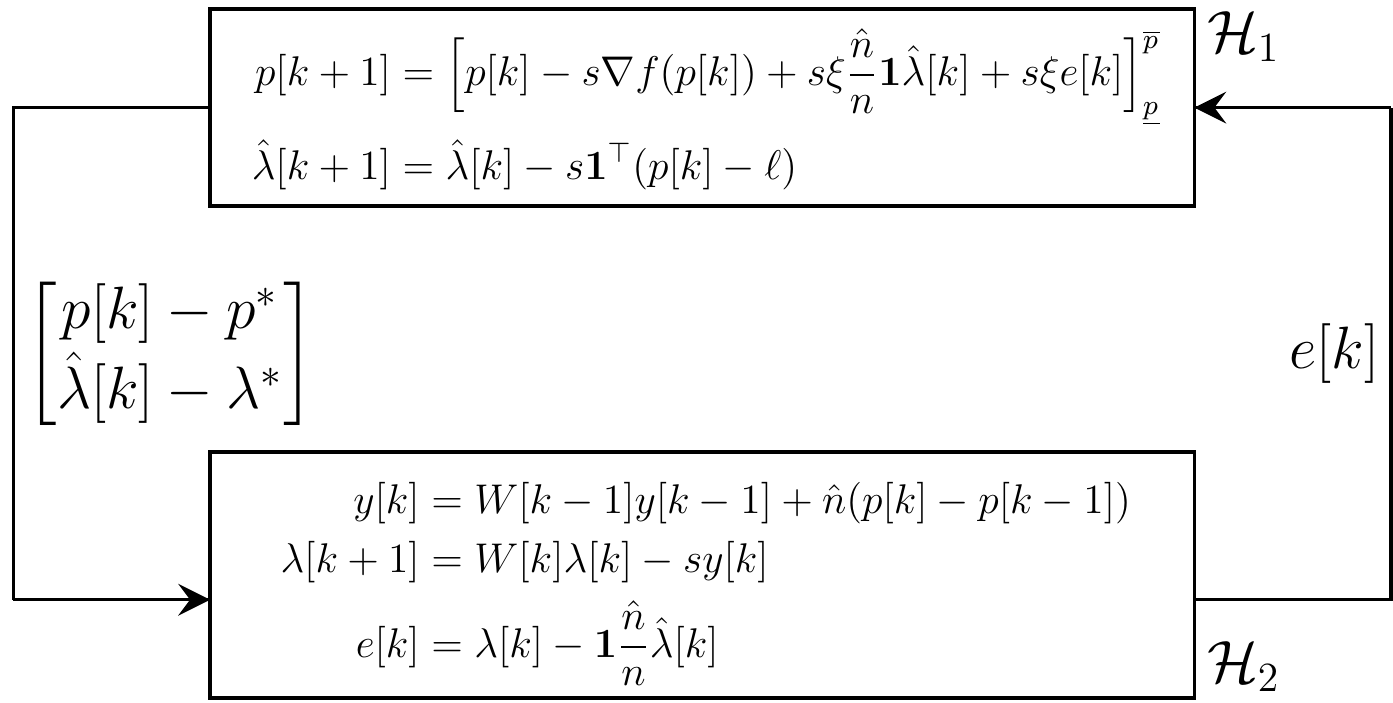} 
    \caption{Algorithm~\eqref{primal_dual_un} as a feedback system.}
    \vspace{-5pt} 
    \label{fig:control_interpret}
\end{figure}
Finding the relationship between the loop gain of the feedback system and the step-size~$s$ allows us to quantify the effect of the feedback system on the convergence error in terms of the step-size. We later show that the loop gain can be decreased by decreasing~$s$. As a matter of fact, if the loop gain is sufficiently small, then, the feedback loop does not amplify the energy of the convergence error, and, on the contrary, the error eventually decays to zero, which follows from the small-gain theorem.

\subsection{Convergence Analysis}\label{subsec:conv_analysis_un}
In order to invoke the small-gain theorem, we must first show that the following relations between the energy of the convergence error and that of the disturbance hold:
\begin{itemize}
\item[{\textbf{R1.}}] $\|z\|_2^{a,K}\leq \alpha_1\|e\|_2^{a,K}+\beta_1$ for some positive~$\alpha_1$ and $\beta_1$,
\item[{\textbf{R2.}}] $\|e\|_2^{a,K}\leq s\alpha_2\|z\|_2^{a,K}+\beta_2$ for some positive~$\alpha_2$ and $\beta_2$,
\end{itemize}
for some $a \in (0,1)$, sufficiently small~$s>0$, and $\forall \xi \in (0,\frac{n}{\hat{n}}]$, where \[z[k] \coloneqq \begin{bmatrix}p[k]-p^*\\ \hat{\lambda}[k]-\lambda^* \end{bmatrix}\]
denotes the convergence error.
The results {\textbf{R1}} and {\textbf{R2}} are equivalent to ensuring that the systems~$\mathcal{H}_1$ and $\mathcal{H}_2$ in Fig.~\ref{fig:control_interpret} are finite-gain stable.
From {\textbf{R1}} and {\textbf{R2}}, it can be determined that the loop gain is $s\alpha_1\alpha_2$.
Noticing that the gain~$s\alpha_1\alpha_2$ becomes strictly smaller than $1$ for sufficiently small~$s$, we later show that $\|z\|_2^{a,K}$ becomes bounded for all $K>0$, and that $z[k]$ converges to zero at a geometric rate $\mathcal{O}(a^k)$. 
In the following, we show that the relations {\textbf{R1}} and {\textbf{R2}} hold and present the convergence results for algorithm~\eqref{primal_dual_un}.

In the next result, we establish that $\mathcal{H}_1$ is finite-gain stable.
\begin{proposition}\label{prop:H1_un}
Let Assumption~\ref{objective_assumption} hold.
Then, under~\eqref{H1}, we have that
\begin{align}
\textnormal{\textbf{R1. }} \|z\|_2^{a,K}\leq \alpha_1\|e\|_2^{a,K}+\beta_1,
\label{z_to_e_un}
\end{align}
for some positive~$\alpha_1$ and $\beta_1$, $a \in (0,1)$, sufficiently small $s>0$, and $\forall \xi \in (0,\frac{n}{\hat{n}}]$.
\end{proposition}
\begin{proof}
Letting
\begin{align*}
G[k] &\coloneqq p[k]-s\nabla f(p[k])+s\xi\frac{\hat{n}}{n}\mathbf{1}\hat{\lambda}[k]+s\xi e[k],\\
H[k] &\coloneqq \hat{\lambda}[k]-s\mathbf{1}^\top (p[k]-\ell),\\
G(p^*,\lambda^*) &\coloneqq p^*-s\nabla f(p^*)+s\xi\frac{\hat{n}}{n}\mathbf{1}\lambda^*,
\end{align*}
we note that the following relations clearly hold:
\begin{subequations}
\begin{align*}
p[k+1]&=\big[G[k]\big]_{\underline{p}}^{\overline{p}}, \mbox{ }p^* = \big[G(p^*,\lambda^*)\big]_{\underline{p}}^{\overline{p}}, \mbox{ } \hat{\lambda}[k+1] = H[k].
\end{align*}
\end{subequations}
Then, by the Projection Theorem \cite[Proposition~2.1.3]{NonlinearProgramming}, we have that
\begin{align}
\|z[k+1]\|&=\left\|\begin{bmatrix}p[k+1]-p^*\\ \hat{\lambda}[k+1]-\lambda^* \end{bmatrix}\right\|\nonumber\\&\leq \left\|\begin{bmatrix}G[k]\\ H[k] \end{bmatrix}-\begin{bmatrix}G(p^*,\lambda^*)\\ \lambda^* \end{bmatrix}\right\|.
\label{projection}
\end{align}
Next, it follows from the mean value theorem \cite[Theorem~5.1]{Rudin} applied to each component in $\nabla f(p[k]) - \nabla f(p^*)$ that
\begin{align}
\nabla f(p[k]) - \nabla f(p^*) = \nabla^2f(\upsilon[k])(p[k]-p^*), \label{mean-value-th}
\end{align}
where $\upsilon[k]\coloneqq[\upsilon_1[k],\upsilon_2[k],\dots,\upsilon_n[k]]^\top$, with $\upsilon_i[k]$ lying on the line segment connecting $p_i[k]$ and $p_i^*$, and $\nabla^2f(\upsilon[k])$ is the Hessian of $f(\mathrm{x})$ at $\mathrm{x}=\upsilon[k]$. Then, by using~\eqref{mean-value-th}, we have that
\begin{align}
\begin{bmatrix}G[k]\\ H[k] \end{bmatrix}-\begin{bmatrix}G(p^*,\lambda^*)\\ \lambda^* \end{bmatrix}&=A[k]\begin{bmatrix}p[k]-p^*\\ \hat{\lambda}[k]-\lambda^* \end{bmatrix}+s\xi\begin{bmatrix}e[k]\\0 \end{bmatrix},
\label{ave_dyn2}
\end{align}
where
\begin{equation*}
A[k] \coloneqq \begin{bmatrix}I - s\nabla^2f(\upsilon[k])&s\xi\frac{\hat{n}}{n}\mathbf{1}\\-s\mathbf{1}^\top&1 \end{bmatrix}.
\end{equation*}
Define
\begin{equation*}
B[k] \coloneqq \begin{bmatrix}\nabla^2f(\upsilon[k])&-\xi\frac{\hat{n}}{n}\mathbf{1}\\\mathbf{1}^\top&0 \end{bmatrix}
\end{equation*}
so that $A[k]=I-sB[k]$. We show that all eigenvalues of~$B[k]$ have a strictly positive real part. Suppose $\mu$ is an eigenvalue of~$B[k]$ and $[v^\mathrm{H},w^\mathrm{H}]^\mathrm{H}$ is an eigenvector corresponding to $\mu$, where $x^\mathrm{H}$ denotes the Hermitian transpose of~$x$. Then, on the one hand, we have that
\begin{align*}
\operatorname{Re}\left([v^\mathrm{H},w^\mathrm{H}]B[k]\begin{bmatrix}v\\ w \end{bmatrix}\right)&=\operatorname{Re}\left(\mu[v^\mathrm{H},w^\mathrm{H}]\begin{bmatrix}v\\ w \end{bmatrix}\right)\nonumber\\&=\operatorname{Re}(\mu)(\|v\|_2^2+\|w\|_2^2).
\end{align*}
On the other hand, we have that
\begin{align}
\operatorname{Re}\left([v^\mathrm{H},w^\mathrm{H}]B[k]\begin{bmatrix}v\\ w \end{bmatrix}\right)&=\operatorname{Re}\big(v^\mathrm{H}\nabla^2f(\upsilon[k])v\nonumber\\&\quad- \xi\frac{\hat{n}}{n} v^\mathrm{H}\mathbf{1}w+w^\mathrm{H}\mathbf{1}^\top v\big).\label{proof_ineq1}
\end{align}
From the fact that
\[B[k]\begin{bmatrix}v\\w\end{bmatrix} = \mu \begin{bmatrix}v\\w\end{bmatrix},\] we have that $\mathbf{1}^\top v = \mu w$, from where it follows that \[w^\mathrm{H}\mathbf{1}^\top v = \mu \|w\|_2^2.\]
Therefore, $v^\mathrm{H}\mathbf{1}w = w^\mathrm{H}\mathbf{1}^\top v =  \mu \|w\|_2^2\geq 0$, from where it follows that
\begin{align}
\operatorname{Re}\big(-\xi\frac{\hat{n}}{n} v^\mathrm{H}\mathbf{1}w+w^\mathrm{H}\mathbf{1}^\top v\big)&=(1-\xi\frac{\hat{n}}{n})\mu \|w\|_2^2\nonumber\\&\geq 0,
\label{proof_ineq2}
\end{align}
because $\xi\in(0,\frac{n}{\hat{n}}]$. By applying~\eqref{proof_ineq2} to~\eqref{proof_ineq1} and using Assumption~\ref{objective_assumption}, we obtain that
\begin{align}
\operatorname{Re}\left([v^\mathrm{H},w^\mathrm{H}]B[k]\begin{bmatrix}v\\ w \end{bmatrix}\right)&\geq\operatorname{Re}\big(v^\mathrm{H}\nabla^2f(\upsilon[k])v\big)\geq m\|v\|_2^2\nonumber\\&>0,\label{proof_ineq3}
\end{align}
$\forall v\neq 0$. 
If $\operatorname{Re}(\mu)=0$, then, it follows from~\eqref{proof_ineq3} that $v=0$, and 
\begin{equation*}
B[k]\begin{bmatrix}0\\ w \end{bmatrix} = 0,
\end{equation*}
from which we conclude that $\mathbf{1}w = 0$, and $w=0$, contradicting the fact that $[v^\mathrm{H},w^\mathrm{H}]^\mathrm{H} \neq 0$. Therefore, all eigenvalues of~$B[k]$ have a strictly positive real part, and, for sufficiently small~$s$, the spectral radius of~$A[k]$ denoted by~$\rho(A[k])$ is strictly less than~$1$.

In the following, we show that there exists an induced matrix norm~$\|\cdot\|$ such that $\|A[k]\|\leq\gamma$, for some $\gamma<1$, $\forall k$.
Let
\begin{align*}
A[k] = A + s\tilde{A}[k],
\end{align*}
where
\begin{align*}
A &= \begin{bmatrix}I - s\nabla^2f(\underline{p})&s\xi\frac{\hat{n}}{n}\mathbf{1}\\-s\mathbf{1}^\top&1 \end{bmatrix},\\
\tilde{A}[k] &= \begin{bmatrix}\nabla^2f(\underline{p})-\nabla^2f(\upsilon[k])&0\\0&0 \end{bmatrix}.
\end{align*}
By the Schur triangularization theorem (see, e.g., \cite[Theorem~2.3.1]{Horn_Johnson}), there is a unitary matrix~$U$, i.e., $U^\mathrm{H}U=UU^\mathrm{H} = I$, and an upper triangular~$\Lambda$ such that $A = U^\mathrm{H}\Lambda U$, where the diagonal entries of~$\Lambda$ are the eigenvalues of~$A$.
In the following, we use the fact that if $\vertiii{A}$ is a matrix norm, then, $\vertiii{S^{-1}AS}$ is also a matrix norm, for any real matrix~$A$ and non-singular~$S$ (see, e.g., \cite[Theorem~5.6.7]{Horn_Johnson}). Letting $D_t \coloneqq \diag(t,t^2,\dots,t^{n+1})$, we choose the following matrix norm:
\begin{align*}
\vertiii{A[k]}&\coloneqq \|(U^\mathrm{H}D_t^{-1})^{-1}A[k]U^\mathrm{H}D_t^{-1}\|_1\nonumber\\&= \|D_tUA[k]U^\mathrm{H}D_t^{-1}\|_1\nonumber\\&=\|D_tUU^\mathrm{H}\Lambda UU^\mathrm{H}D_t^{-1} + sD_tU\tilde{A}[k]U^\mathrm{H}D_t^{-1}\|_1
\nonumber\\&=\|D_t\Lambda D_t^{-1} + sD_tU\tilde{A}[k]U^\mathrm{H}D_t^{-1}\|_1,
\end{align*}
where $\|[a_{ij}]\|_1\coloneqq\max\limits_{j}\sum_{i=1}^n |a_{ij}|$, and $D_t\Lambda D_t^{-1}$ is given by
\begin{align*}
D_t\Lambda D_t^{-1} = 
\begin{bmatrix}
\Lambda_{11} & t^{-1}\Lambda_{12} & t^{-2}\Lambda_{13} & \ldots & t^{-n}\Lambda_{1,n+1}\\
0 & \Lambda_{22} & t^{-1}\Lambda_{23} & \ldots & t^{1-n}\Lambda_{2,n+1}\\
&&&\ddots&\\
0 & 0 & 0 & \ldots & \Lambda_{n+1,n+1}
\end{bmatrix}.
\end{align*}
\newcounter{tempequationcounter}
Letting $X[k] \coloneqq D_t\Lambda D_t^{-1} + sD_tU\tilde{A}[k]U^\mathrm{H}D_t^{-1}$, and $Z[k] \coloneqq D_tU\tilde{A}[k]U^\mathrm{H}D_t^{-1}$, we compute $X[k]$ in~\eqref{proof-mat-norm3} (see the top of the next page).
\begin{figure*}[!t]
\setcounter{tempequationcounter}{\value{equation}}
\begin{align}\label{proof-mat-norm3}
X[k] = 
\begin{bmatrix}
\Lambda_{11}+sZ_{11}[k] & t^{-1}(\Lambda_{12}+sZ_{12}[k]) & \ldots & t^{-n}(\Lambda_{1,n+1}+sZ_{1,n+1}[k])\\
stZ_{21}[k] & \Lambda_{22}+sZ_{22}[k] & \ldots & t^{1-n}(\Lambda_{2,n+1}+sZ_{2,n+1}[k])\\
\vdots&\vdots&\ddots&\vdots\\
st^{n}Z_{n+1,1}[k] & st^{n-1}Z_{n+1,2}[k] & \ldots & \Lambda_{n+1,n+1}+sZ_{n+1,n+1}[k]
\end{bmatrix}.
\end{align}%
\end{figure*}
From~\eqref{proof-mat-norm3}, we have that
\begin{align*}
\sum_i X_{ij}[k] = \Lambda_{jj}+\sum_{i< j} \Lambda_{ij}t^{i-j} +\sum_{i=1}^{n+1} sZ_{ij}[k]t^{i-j},
\end{align*}
for $j=1,2,\dots,n+1$. Since $\Lambda_{jj}$ is an eigenvalue of~$A$, $\Lambda_{jj}<1$. Then, it is not difficult to see that for sufficiently small~$s$ and large~$t$, we have that $\left|\sum_i X_{ij}[k]\right|\leq\gamma$, $\forall j$, for some $\gamma<1$ and all $k\geq0$. Therefore, $\vertiii{A[k]}\leq \gamma$, $\forall k$.
By \cite[Theorem~5.6.26]{Horn_Johnson}, there is an induced matrix norm~$\|\cdot\|$ such that $\|A[k]\| \leq \vertiii{A[k]}$. Hence, $\|A[k]\|\leq \gamma$, $\forall k$.

Taking $\|\cdot\|$ on both sides of~\eqref{ave_dyn2} and applying the triangle inequality yields
\begin{align}
\left\|\begin{bmatrix}G[k]\\ H[k] \end{bmatrix}-\begin{bmatrix}G(p^*,\lambda^*)\\ \lambda^* \end{bmatrix}\right\|&\leq\|A[k]\|\left\|\begin{bmatrix}p[k]-p^*\\ \hat{\lambda}[k]-\lambda^* \end{bmatrix}\right\|\nonumber\\&\quad+s\xi\left\|\begin{bmatrix}e[k]\\0 \end{bmatrix}\right\|.\label{proof_ineq4}
\end{align}
By applying~\eqref{proof_ineq4} to~\eqref{projection} and using the fact that $\|A[k]\|\leq \gamma$, $\forall k$, we obtain that
\begin{align}
\|z[k+1]\|&=\left\|\begin{bmatrix}p[k+1]-p^*\\ \hat{\lambda}[k+1]-\lambda^* \end{bmatrix}\right\|\nonumber\\
&\leq\|A[k]\|\left\|\begin{bmatrix}p[k]-p^*\\ \hat{\lambda}[k]-\lambda^* \end{bmatrix}\right\|+s\xi\|e[k]\|\nonumber\\
&=\|A[k]\|\|z[k]\|+s\xi\|e[k]\|\nonumber\\
&\leq \gamma\|z[k]\|+s\xi\|e[k]\|,
\label{ave_dyn_ineq}
\end{align}
where we recall that \[z[k] = \begin{bmatrix}p[k]-p^*\\ \hat{\lambda}[k]-\lambda^* \end{bmatrix}.\]
Now, by multiplying both sides of~\eqref{ave_dyn_ineq} by~$a^{-(k+1)}$, we obtain
\begin{align}
a^{-(k+1)}\|z[k+1]\|\leq \frac{\gamma}{a} a^{-k}\|z[k]\|+\frac{s\xi}{a}a^{-k}\|e[k]\|.
\label{ave_dyn_ineq3}
\end{align}
Then, by taking $\max\limits_{0\leq k\leq K}(\cdot)$ on both sides of~\eqref{ave_dyn_ineq3}, we obtain
\begin{align}
&\max\limits_{0\leq k\leq K}a^{-(k+1)}\|z[k+1]\|\leq \frac{\gamma}{a} \max\limits_{0\leq k\leq K} a^{-k}\|z[k]\|\nonumber\\&\quad+\frac{s\xi}{a}\max\limits_{0\leq k\leq K} a^{-k}\|e[k]\|\nonumber\\&\leq \frac{\gamma}{a} \max\limits_{0\leq k\leq K+1} a^{-k}\|z[k]\|+\frac{s\xi}{a}\max\limits_{0\leq k\leq K+1} a^{-k}\|e[k]\|,\label{proof_ineq5}
\end{align}
Since \[\max\limits_{0\leq k\leq K}a^{-(k+1)}\|z[k+1]\| = \max\limits_{0\leq k\leq K+1}a^{-k}\|z[k]\|-\|z[0]\|,\] 
the relation~\eqref{proof_ineq5} can be written as
\begin{align}
\|z\|_{a,K+1}\leq \frac{\gamma}{a}\|z\|_{a,K} + \frac{s\xi}{a}\|e\|_{a,K}+\|z[0]\|,\label{z_to_e_un2}
\end{align}
where $\|x\|_{a,K}\coloneqq \max\limits_{0\leq k\leq K}a^{-k}\|x[k]\|$ for a sequence~$\{x[k]\}_{k=0}^{\infty}$. 
Since $\|z\|_{a,K+1}\geq\|z\|_{a,K}$, it follows from~\eqref{z_to_e_un2} that
\begin{align}
\|z\|_{a,K}\leq \frac{\gamma}{a}\|z\|_{a,K} + \frac{s\xi}{a}\|e\|_{a,K}+\|z[0]\|.\label{z_to_e_un3}
\end{align}
Then, after rearranging~\eqref{z_to_e_un3}, we obtain
\begin{align*}
\|z\|_{a,K}\leq \frac{s\xi}{a-\gamma}\|e\|_{a,K}+\frac{a}{a-\gamma}\|z[0]\|.
\end{align*}
Because $\|\cdot\|_2\leq \alpha\|\cdot\|$ and $\|\cdot\|\leq \beta\|\cdot\|_2$ for some~$\alpha$ and $\beta$, we have that $\|z\|_{a,K}\geq {\|z\|_2^{a,K}}/{\alpha}$, $\|e\|_{a,K} \leq \beta\|e\|_2^{a,K}$. Hence,
\[\frac{\|z\|_2^{a,K}}{\alpha}\leq \frac{s\xi\beta}{a-\gamma}\|e\|_2^{a,K}+\frac{a}{a-\gamma}\|z[0]\|,\]
which can be rewritten as 
\[
\|z\|_2^{a,K}\leq \alpha_1\|e\|_2^{a,K}+\beta_1,
\]
where \[\alpha_1 =  \frac{s\xi\alpha\beta}{a-\gamma},\] and \[\beta_1 = \frac{a\alpha}{a-\gamma}\|z[0]\|,\]
yielding~\eqref{z_to_e_un}.
\end{proof}
{We omit the proof of the next result, where we show that system~$\mathcal{H}_2$ is finite-gain stable, since it is analogous to that of a similar result proposed for directed communication graphs in Section~\ref{sec:DOD_dd}.}
\begin{proposition}\label{prop:H2_un}
Let Assumptions~\ref{objective_assumption} and \ref{assume_comm_model_un1} hold.
Then, under~\eqref{H2}, we have that
\begin{align*}
\textnormal{\textbf{R2. }}\|e\|_2^{a,K}\leq s\alpha_2\|z\|_2^{a,K}+\beta_2,
\end{align*}
for some positive~$\alpha_2$ and $\beta_2$, $a \in (0,1)$, and sufficiently small~$s>0$.
\end{proposition}

{Now, we show the convergence of algorithm~\eqref{primal_dual_un} by applying the small-gain theorem to the results in Propositions~\ref{prop:H1_un}--\ref{prop:H2_un}.
\begin{proposition}\label{prop:small_gain_un}
Let Assumptions~\ref{objective_assumption} and \ref{assume_comm_model_un1} hold.
Then, under algorithm~\eqref{primal_dual_un},
\begin{align}
\|z\|_2^{a,K}\leq \beta,
\label{z_relation_un}
\end{align}
for some $a \in (0,1)$, $\beta>0$, sufficiently small~$s>0$, and $\forall \xi \in (0,\frac{n}{\hat{n}}]$. In particular, $(p_i[k],\lambda_i[k])$ converges to $(p_i^*,\lambda^*)$, $\forall i$, at a geometric rate~$\mathcal{O}(a^k)$.
\end{proposition}
\begin{proof}
By using Propositions~\ref{prop:H1_un} and \ref{prop:H2_un}, it follows that
\begin{align*}
\|z\|_2^{a,K}\leq \alpha_1\|e\|_2^{a,K}+\beta_1 \leq \alpha_1(s\alpha_2\|z\|_2^{a,K}+\beta_2)+\beta_1,
\end{align*}
which, after rearranging, results in
\begin{align*}
\|z\|_2^{a,K}\leq \frac{\alpha_1\beta_2+\beta_1}{1-s\alpha_1\alpha_2}\eqqcolon \beta,
\end{align*}
yielding \eqref{z_relation_un}.
Hence, for sufficiently small~$s$, we have that $s\alpha_1\alpha_2<1$, which ensures that $\beta$ is finite.
\end{proof}
}
Finally, we show that $p^*$ is the solution of \eqref{ED}. 
\begin{lemma}\label{lem:p_star}
Consider~$(p^*,\lambda^*)$, namely, the equilibrium of the nominal system~$\mathcal{H}_1$ with $e[k]\equiv 0$, $\forall k$.
Then, $p^*$ is the solution of \eqref{ED}.
\end{lemma}
\begin{proof}
At the equilibrium, we have that
\begin{align*}
p^* &= \big[p^*-s\nabla f(p^*)+s\xi\frac{\hat{n}}{n}\mathbf{1}\lambda^*\big]_{\underline{p}}^{\overline{p}},\\
\lambda^* &= \lambda^* -  \mathbf{1}^\top(p^*-\ell).
\end{align*}
Then, the following relations hold:
\begin{subequations}\label{KKT}
\begin{align}
0 &= \nabla f(p^*)-\xi\frac{\hat{n}}{n}\mathbf{1}\lambda^* + \mu^* - \nu^*,\\
0 &= \mathbf{1}^\top(p^*-\ell),\\
0 &=\mu_i^*(p_i^*-\overline{p}_i),\\
0 &=\nu_i^*(\underline{p}_i-p_i^*),i=1,\dots,n,
\end{align}
\end{subequations}
where $\mu^*=[\mu_1^*,\dots,\mu_n^*]^\top$, with $\mu_i^*\geq 0$, $i=1,\dots,n$, and $\nu^*=[\nu_1^*,\dots,\nu_n^*]^\top$, with $\nu_i^*\geq 0$, $i=1,\dots,n$. Noticing that \eqref{KKT} represents the Karush-Kuhn-Tucker (KKT) conditions for \eqref{ED}, it follows from \cite[Proposition~3.3.1]{NonlinearProgramming} that $p^*$ is the solution of \eqref{ED}.
\end{proof}

\subsection{Numerical Simulations}
\label{subsec:simulations}
Next, we present the numerical results that illustrate the performance of the proposed distributed primal-dual algorithm~\eqref{primal_dual_un} using the IEEE~$39$--bus test system \cite{Matpower}. 
We randomly pick the load demands and generation capacity constraints of the DERs. For each~$i$, we choose~$f_i(p_i) = a_ip_i^2$, where $a_i>0$ is randomly selected. With regard to the communication model, every pair of nodes are connected by a bidirectional communication link if there is an electrical line between them. Communication links are assumed to fail with probability~$0.2$ independently (and independently between different time steps). The weights~$a_{ij}[k]$, $\{i,j\} \in \mathcal{E}^{(0)}$, are picked using the Metropolis rule \cite{XiBo05}, namely, 
\begin{align}
a_{ij}[k] = \left\{\begin{matrix*}[l]\frac{1}{\max(d_i,d_j)}&\{i,j\} \in \mathcal{E}^{(c)}[k],\\0 & \mbox{otherwise.}\end{matrix*}\right.
\end{align}

For convenience, algorithm~\eqref{primal_dual_un} is referred to as~$\mathbf{PD}_1$. We compare its performance with that of algorithm~\eqref{primal_dual_un_ac}, referred to as~$\mathbf{PD}_2$. In the simulations, $\mathbf{PD}_1$ uses a constant stepsize~$s = 0.01$ and $\xi = 0.05$. In contrast, $\mathbf{PD}_2$ needs to use a diminishing stepsize of the form~$s[k] = a/(k+b)$, where $a>0$ and $b>0$, in order to guarantee convergence. Both algorithms are initialized with $p[0]=0$. In Fig.~\ref{fig:num_results_un}, we provide the convergence error, namely, the Euclidean distance between the exact and iterative solutions, $\|p[k]-p^*\|_2$, for both algorithms. It can be seen that $\mathbf{PD}_1$ significantly outperforms $\mathbf{PD}_2$, and has geometric convergence speed. 
\begin{figure}
    \centering
	\includegraphics[trim=0.2cm 0cm 0cm 0cm, clip=true, scale=0.48]{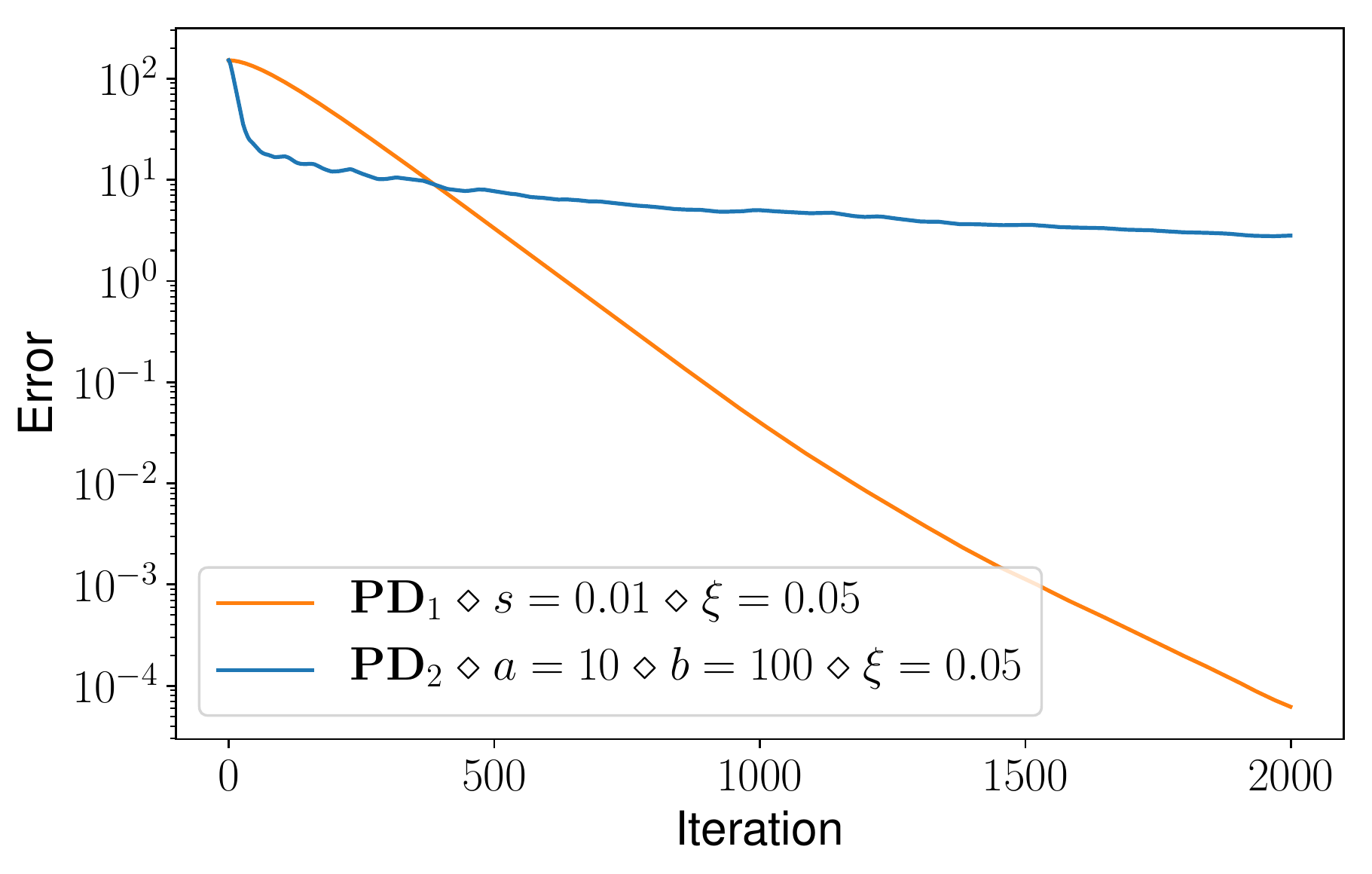} 
    \caption{Trajectory of $\|p[k]-p^*\|_2$ for algorithms $\mathbf{PD}_1$, $\mathbf{PD}_2$.}
    \vspace{-5pt} 
    \label{fig:num_results_un}
\end{figure}

\section{DER Coordination Over Time-Varying Directed Graphs}\label{sec:DOD_dd}
In this section, we present a distributed algorithm for solving the DER coordination problem~\eqref{ED} over time-varying directed communication graphs. 

\subsection{Communication Network Model}\label{subsec:cyber_layer}
Here, we introduce the model describing the communication network that enables the unidirectional exchange of information between DERs. 
Let $\mathcal{G}^{(0)} = (\mathcal{V},\mathcal{E}^{(0)})$ denote a directed graph, where each element in the node set~$\mathcal{V} \coloneqq \{1,2,\dots,n\}$ corresponds to a DER, and $(i,j) \in \mathcal{E}^{(0)}$ if there is a communication link that allows DER at node~$i$ to send information to DER at node~$j$ (but not vice versa).
During any time interval~$(t_k,t_{k+1})$, successful data transmissions among the DERs can be captured by the directed graph~$\mathcal{G}^{(c)}[k]=(\mathcal{V},\mathcal{E}^{(c)}[k])$, where $\mathcal{E}^{(c)}[k]\subseteq \mathcal{E}^{(0)}$ is the set of active communication links, with $(i,j) \in \mathcal{E}^{(c)}[k]$ if node~$j$ receives information from node~$i$ during time interval~$(t_k,t_{k+1})$, but not necessarily vice versa. 
Let $\mathcal{N}_i^+\coloneqq\{j\in\mathcal{V}:(i,j)\in\mathcal{E}^{(0)}\}$ denote the set of nominal out-neighbors, and $d_i^+ \coloneqq |\mathcal{N}_i^+|+1$ denote the nominal out-degree of node~$i$. Let $\mathcal{N}_i^+[k]$ and $\mathcal{N}_i^-[k]$ denote the sets of out-neighbors and in-neighbors of node~$i$, respectively, during time interval~$(t_k,t_{k+1})$, i.e., $\mathcal{N}_i^+[k]\coloneqq\{j\in\mathcal{V}:(i,j) \in \mathcal{E}^{(c)}[k]\}$ and $\mathcal{N}_i^-[k]\coloneqq\{\ell\in\mathcal{V}:(\ell,i) \in \mathcal{E}^{(c)}[k]\}$. We define node~$i$'s instantaneous (communication) out-degree (including itself) to be $D_i^+[k] \coloneqq |\mathcal{N}_i^+[k]|+1$. We make the following two standard assumptions (see, e.g., \cite{Nedic15,Nedic17}). 
\begin{assumption}
\label{assume_comm_model1}
There exists some positive integer~$B$ such that the graph with node set~$\mathcal{V}$ and edge set~$\bigcup_{l= kB}^{(k+1)B-1}\mathcal{E}^{(c)}[l]$ is strongly connected for $k=0,1,\dots$. 
\end{assumption}
\begin{assumption}
\label{assume_comm_model2}
The value of~$D_i^+[k]$ is known to node~$i$, $i=1,2,\dots,n$, for all $k\geq 0$.
\end{assumption}
In Assumption~\ref{assume_comm_model2}, we assume that each node knows its instantaneous out-degree, which, in practice, may not be readily available. However, this assumption will be relaxed later in Section~\ref{sec:robust-DOD-dd}.

\subsection{Ratio Consensus Algorithm}
We begin with a brief overview of the ratio consensus algorithm (see, e.g., \cite{DoHa12}) utilized to estimate the average power imbalance (or total power imbalance if $n$ is known) and to update the Lagrange multipliers, $\lambda_i[k]$, in the distributed algorithms proposed later. 

Consider a group of nodes indexed by the set~$\mathcal{V}$, each with some real initial value, i.e., $v_i$ at node~$i$. Each node aims to obtain the average of the initial values via exchange of information over the graph~$\mathcal{G}^{(c)}[k]$. To this end, we let node~$i$ maintain two variables, $\mu_i[k]$ and $\nu_i[k]$ such that $\mu_i[0] = v_i$ and $\nu_i[0] = 1$. We first consider the following updates performed by node~$i$:
\begin{subequations}\label{RatioCons}
\begin{align}
\mu_{i}[k+1] &= \sum_{j\in\mathcal{N}_i^-[k]\cup\{i\}}\frac{\mu_j[k]}{D_j^+[k]},\label{mu_ratio_cons}\\
\nu_{i}[k+1] &= \sum_{j\in\mathcal{N}_i^-[k]\cup\{i\}}\frac{\nu_j[k]}{D_j^+[k]},\label{nu_ratio_cons}\\
r_i[k+1] &= \frac{\mu_i[k+1]}{\nu_i[k+1]}.
\end{align}
\end{subequations}
We write~\eqref{mu_ratio_cons}--\eqref{nu_ratio_cons} in a matrix-vector form as follows:
\begin{subequations}\label{VectRatioCons}
\begin{align}
\mu[k+1] &= P[k]\mu[k],\\
\nu[k+1] &= P[k]\nu[k],
\end{align}
\end{subequations}
where $P[k]$ is an $n\times n$ matrix with $P_{ij}[k] = 1/D_j^+[k]$, $j\in\mathcal{N}_i^-[k]\cup\{i\}$, and $P_{ij}[k] = 0$, otherwise.
We note that $P[k]$ is column stochastic, and it can be shown that $r_i[k]$ converges to the average of the initial values, namely, $\lim_{k\rightarrow \infty}r_i[k] = \frac{\sum_{i}v_i}{n}$ \cite{DoHa12}, as long as Assumption~\ref{assume_comm_model1} holds. 
\subsection{Distributed Primal-Dual Algorithm}
We propose the following distributed primal-dual algorithm, which is based on the ratio consensus algorithm~\eqref{RatioCons}, where each node~$i$ runs the following iterations:
\begin{subequations}\label{primal_dual_alg}
\begin{align}
p_i[k+1]&=\Big[p_i[k]-sf^\prime_i(p_i[k])+s\xi x_i[k]\Big]_{\underline{p}_i}^{\overline{p}_i},\label{u_update}\\
\lambda_i[k+1] &= \sum_{j\in\mathcal{N}_i^-[k]\cup\{i\}}\frac{\lambda_j[k]-sy_j[k]}{D_j^+[k]},\label{lmbd_update_dd}\\
v_i[k+1]&=\sum_{j\in\mathcal{N}_i^-[k]\cup\{i\}}\frac{v_j[k]}{D_j^+[k]},\\
x_i[k+1]&=\frac{\lambda_i[k+1]}{v_i[k+1]},\\
y_i[k+1]&=\sum_{j\in\mathcal{N}_i^-[k]\cup\{i\}}\frac{y_j[k]}{D_j^+[k]}+\hat{n}(p_i[k+1]-p_i[k]),
\end{align}
\end{subequations}
where $y_i[k]$ is the estimate at instant~$k$ of the total power imbalance, $\mathbf{1}^\top (p[k]-\ell)$, at node~$i$.
The iterations in \eqref{primal_dual_alg} are initialized with $x_i[0] = 0$, $\lambda_i[0] = 0$, $v_i[0]=1$, and $y_i[0] = \hat{n}(p_i[0]-\ell_i)$. We note that the iterations used to update $\lambda_i[k]$, $v_i[k]$ and $y_i[k]$ are similar to the so-called Push-DIGing algorithm proposed in \cite{Nedic17} for solving an unconstrained consensus optimization problem.
{\begin{remark}
Algorithm~\eqref{primal_dual_alg} is similar to the algorithm in \cite{YaWu19} in that they both use the Push-DIGing algorithm and the gradient tracking idea in \cite{Nedic17} to update the dual variables, and both require agents to know their instantaneous communication out-degrees~$D_j^+[k]$'s at each iteration~$k$.
However, there are a few subtle differences that can be pointed out. Algorithm~\eqref{primal_dual_alg} is based on the first order Lagrangian method \cite{NonlinearProgramming}, while the algorithm in \cite{YaWu19} is based on the dual-ascent method \cite{Boyd11}. Algorithm~\eqref{primal_dual_alg} also contains additional parameters, $\xi$ and $\hat{n}$, that can noticeably improve performance if they are carefully tuned. But more importantly, algorithm~\eqref{primal_dual_alg} serves as an intermediate step for designing a robustified extension (to be proposed later) that converges geometrically fast even when the instantaneous out-degrees are unknown to the agents, whereas the algorithms in \cite{DuYa18} and \cite{YaWu19} require the agents to know their instantaneous out-degrees.
\end{remark}}


As in the case of undirected graphs, algorithm~\eqref{primal_dual_alg} can be represented as a feedback interconnection of a nominal system, denoted by $\vec{\mathcal{H}}_1$, and a disturbance system, denoted by $\vec{\mathcal{H}}_2$. 
To this end, let $e[k] \coloneqq x[k]-\big(\frac{1}{n}\mathbf{1}^\top\lambda[k]\big)\mathbf{1}$, and $\hat{\lambda}[k] \coloneqq \frac{1}{\hat{n}}\mathbf{1}^\top\lambda[k]$; then, we  define the nominal system, $\vec{\mathcal{H}}_1$, as follows:
\begin{subequations}\label{H1_dd}
\begin{numcases}{\vec{\mathcal{H}}_1:}
\begin{aligned}p[k+1]&=\Big[p[k]-s\nabla f(p[k])+s\xi\frac{\hat{n}}{n}\mathbf{1}\hat{\lambda}[k]\\&\quad+s\xi e[k]\Big]_{\underline{p}}^{\overline{p}},\end{aligned}\label{H1_dd_1}\\
\hat{\lambda}[k+1] = \hat{\lambda}[k]-s\mathbf{1}^\top (p[k]-\ell),\label{H1_dd_2}
\end{numcases}
\end{subequations} 
where we substituted $e[k]+\frac{\hat{n}}{n}\mathbf{1}\hat{\lambda}[k]$ for $x[k]$ in~\eqref{u_update} to obtain~\eqref{H1_dd_1}, and we summed~\eqref{lmbd_update_dd} over all~$i$ and divided the result by~$\hat{n}$ to obtain~\eqref{H1_dd_2}.
We note that $e[k]$ is the vector of deviations from their average at instant~$k$ of the local estimates of the Lagrange multiplier; without~$e[k]$, the nominal system~$\vec{\mathcal{H}}_1$ has almost the same form as~\eqref{central_alg}. In fact, the nominal systems for the undirected and directed cases are the same, namely, $\vec{\mathcal{H}}_1=\mathcal{H}_1$.
Now, we define the disturbance system, $\vec{\mathcal{H}}_2$, as follows: 
\begin{subequations}\label{H2_dd}
\begin{numcases}{\vec{\mathcal{H}}_2:}
\begin{aligned}y[k]&=P[k-1]y[k-1]\\&\quad+\hat{n}(p[k]-p[k-1]),\end{aligned}\label{H2_dd_y}\\
\lambda[k+1]= P[k](\lambda[k]-sy[k]),\label{H2_dd_lmbd}\\
v[k+1]= P[k]v[k],\\
x[k+1]=(V[k+1])^{-1}\lambda[k+1],\label{H2_dd_x}\\
e[k] = x[k]-\Big(\frac{\hat{n}}{n}\hat{\lambda}[k]\Big)\mathbf{1},
\end{numcases}
\end{subequations} 
where $V[k] \coloneqq \diag(v[k])$, $P[k]\in\mathds{R}^{n\times n}$ with $P_{ii}[k] \coloneqq {1}/({D_i^+[k]})$ and
\begin{align*}
P_{ij}[k] \coloneqq \left\{\begin{array}{l l} \frac{1}{D_j^+[k]} & \mbox{if }(j,i) \in \mathcal{E}^{(c)}[k],\\0 & \mbox{else}.\end{array}\right.
\end{align*}
Then, as illustrated in Fig.~\ref{fig:control_interpret_dd}, algorithm~\eqref{primal_dual_alg} can be viewed as a feedback interconnection of~$\vec{\mathcal{H}}_1$ and $\vec{\mathcal{H}}_2$,
where $(p^*,\lambda^*)$ is the equilibrium of~\eqref{H1_dd} when $e[k]\equiv 0$, $\forall k$.
\begin{figure}
    \centering 
	\includegraphics[trim=0cm 0cm 0cm 0cm, clip=true, scale=0.6]{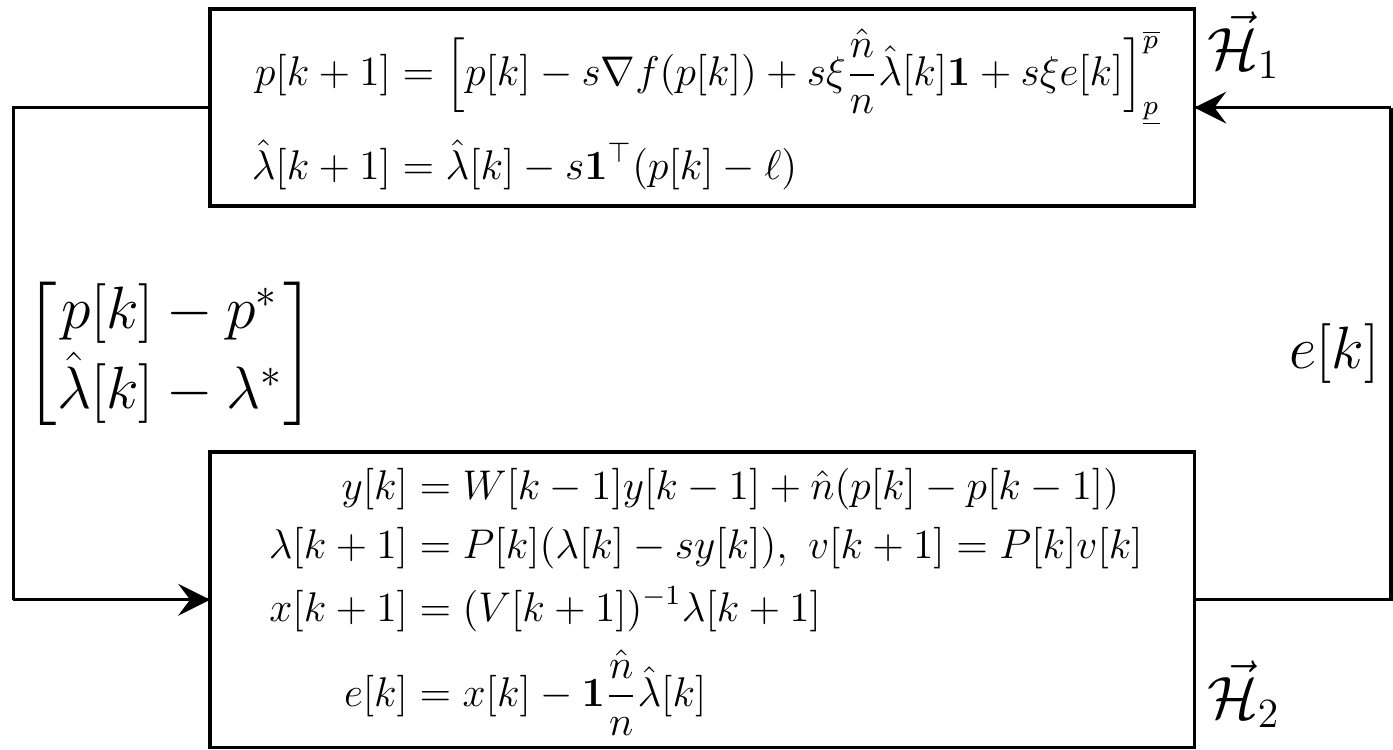} 
    \caption{Algorithm~\eqref{primal_dual_alg} as a feedback system.}
    \vspace{-5pt} 
    \label{fig:control_interpret_dd}
\end{figure}

\subsection{Convergence Analysis}\label{subsec:conv_analysis_dd}
{
As in the case of undirected graphs, we establish that the same relations {\textbf{R1}} and {\textbf{R2}} hold for directed graphs, namely, that $\vec{\mathcal{H}}_1$ and $\vec{\mathcal{H}}_2$ are finite-gain stable. This allows us to apply the small-gain theorem to prove the convergence of~\eqref{primal_dual_alg}.
The proof of the first result, where we show that $\vec{\mathcal{H}}_1$ is finite-gain stable, is omitted because $\vec{\mathcal{H}}_1=\mathcal{H}_1$, and, hence, it is identical to that of Proposition~\ref{prop:H1_un}.}
\begin{proposition}\label{prop:H1_dd}
Let Assumption~\ref{objective_assumption} hold.
Then, under~\eqref{H1_dd}, we have that
\begin{align*}
\textnormal{\textbf{R1. }} \|z\|_2^{a,K}\leq \alpha_1\|e\|_2^{a,K}+\beta_1,
\end{align*}
for some positive~$\alpha_1$ and $\beta_1$, $a \in (0,1)$, sufficiently small~$s>0$, and $\forall\xi \in (0,\frac{n}{\hat{n}}]$,
where \[z[k] \coloneqq \begin{bmatrix}p[k]-p^*\\ \hat{\lambda}[k]-\lambda^* \end{bmatrix}.\]
\end{proposition}
In the next result, we show that $\vec{\mathcal{H}}_2$ is finite-gain stable.
\begin{proposition}\label{prop:H2_dd}
Let Assumptions~\ref{objective_assumption}, \ref{assume_comm_model1}, and \ref{assume_comm_model2} hold.
Then, under~\eqref{H2_dd}, we have that
\begin{align}
\textnormal{\textbf{R2. }}\|e\|_2^{a,K}\leq s\alpha_2\|z\|_2^{a,K}+\beta_2,\label{e_to_z}
\end{align}
for some positive~$\alpha_2$ and $\beta_2$, $a \in (0,1)$, and sufficiently small~$s>0$.
\end{proposition}
\begin{proof}
Let $V[k] \coloneqq \diag(v[k])$, and $P[k]\in\mathds{R}^{n\times n}$ with $P_{ii}[k] \coloneqq {1}/({D_i^+[k]})$ and
\begin{align*}
P_{ij}[k] \coloneqq \left\{\begin{array}{l l} \frac{1}{D_j^+[k]} & \mbox{if }(j,i) \in \mathcal{E}^{(c)}[k],\\0 & \mbox{else}.\end{array}\right.
\end{align*}
Letting $h[k]\coloneqq (V[k])^{-1}y[k]$, and $R[k]\coloneqq (V[k+1])^{-1}P[k]V[k]$, we rewrite~\eqref{H2_dd_y} and \eqref{H2_dd_x} in vector form as follows:
\begin{subequations}\label{primal_dual_alg_vector}
\begin{align}
x[k+1]&=(V[k+1])^{-1}P[k](V[k]x[k]-sy[k]) \nonumber\\&= R[k](x[k]-sh[k]),\\
h[k+1]&=(V[k+1])^{-1}y[k+1]\nonumber\\&=(V[k+1])^{-1}(P[k]V[k]h[k]\nonumber\\&\quad+\hat{n}(p[k+1]-p[k]))\nonumber\\&=R[k]h[k]+(V[k+1])^{-1}(\hat{n}(p[k+1]-p[k])),
\end{align}
\end{subequations}
Letting $\delta[k+1] \coloneqq \hat{n}(p[k+1]-p[k])$,
and using the triangle inequality, we obtain that
\begin{align}
\|\delta[k+1]\|_2 &= \hat{n}\|p[k+1]-p^*-p[k]+p^*\|_2\nonumber\\&\leq \hat{n}(\|p[k+1]-p^*\|_2+\|p[k]-p^*\|_2)\nonumber\\&\leq \hat{n}(\|z[k+1]\|_2+\|z[k]\|_2).
\label{delta_to_z_ineq}
\end{align}
Then, by taking $\max\limits_{0\leq k\leq K}(\cdot)$ on both sides of~\eqref{delta_to_z_ineq}, we obtain 
\begin{align}
\max\limits_{0\leq k\leq K}\|\delta[k+1]\|_2 &\leq \hat{n}\Big(\max\limits_{0\leq k\leq K}\|z[k+1]\|_2\nonumber\\&\quad+\max\limits_{0\leq k\leq K}\|z[k]\|_2\Big),\nonumber\\
&\leq 2\hat{n}\max\limits_{0\leq k\leq K+1}\|z[k]\|_2.
\label{delta_to_z0}
\end{align}
Since $\max\limits_{0\leq k\leq K}\|\delta[k+1]\|_2 = \max\limits_{0\leq k\leq K+1}\|\delta[k]\|_2 - \|\delta[0]\|_2$, it follows from~\eqref{delta_to_z0} that
\begin{align}
\|\delta\|_2^{a,K} \leq 2\hat{n}\|z\|_2^{a,K} + \|\delta[0]\|_2.
\label{delta_to_z}
\end{align}
Let $\tilde{x}[k] \coloneqq x[k]-\big(\frac{1}{n}\mathbf{1}^\top x[k]\big)\mathbf{1}$, and $\tilde{h}[k] \coloneqq h[k]-\big(\frac{1}{n}\mathbf{1}^\top h[k]\big)\mathbf{1}$. 
We recall that, by definition, $e[k] \coloneqq x[k]-\big(\frac{1}{n}\mathbf{1}^\top\lambda[k]\big)\mathbf{1}$.
Then, by using the triangle inequality and noting that $\lambda[k] = V[k]x[k]$, we have that
\begin{align}
\|e[k]\|_2 &= \Big\|x[k]-\Big(\frac{1}{n}\mathbf{1}^\top\lambda[k]\Big)\mathbf{1}\Big\|_2\nonumber\\
&\leq \Big\|\Big(\frac{1}{n}\mathbf{1}^\top x[k]\Big)\mathbf{1}-\Big(\frac{1}{n}\mathbf{1}^\top\lambda[k]\Big)\mathbf{1}\Big\|_2+\|\tilde{x}[k]\|_2 \nonumber\\&
= \Big\|\frac{1}{n}\big(\mathbf{1}^\top x[k]-\mathbf{1}^\top V[k]x[k]\big)\mathbf{1}\Big\|_2 +\|\tilde{x}[k]\|_2
\nonumber\\& = \frac{1}{n}\|(\mathbf{1}-v[k])^\top x[k])\mathbf{1}\|_2+\|\tilde{x}[k]\|_2
\nonumber\\& = \frac{1}{\sqrt{n}}\|(\mathbf{1}-v[k])^\top x[k])\|_2  +\|\tilde{x}[k]\|_2
\nonumber\\& = \frac{1}{\sqrt{n}}\|(\mathbf{1}-v[k])^\top (I-\frac{1}{n}\mathbf{1}\mathbf{1}^\top) x[k])\|_2+\|\tilde{x}[k]\|_2\label{e_to_x_0}\\
& = \frac{1}{\sqrt{n}}\|(\mathbf{1}-v[k])\|_2\|\tilde{x}[k]\|_2+\|\tilde{x}[k]\|_2\leq 2\|\tilde{x}[k]\|_2,
\label{e_to_x}
\end{align}
where using the fact that $v[k]^\top \mathbf{1} = n$ yields~\eqref{e_to_x_0}, and the fact that $0\leq v_i[k]\leq 1$ yields~\eqref{e_to_x}.
For further analysis, we invoke the following results \cite[Lemmas~15, 16]{Nedic17}:
\begin{lemma}
\begin{align}
\|\tilde{h}\|_2^{a,K} \leq \gamma_1\|\delta\|_2^{a,K}+\gamma_2,
\label{h_to_z}
\end{align}
for some $\gamma_1$ and $\gamma_2$. [Precise values can be found in \cite[Lemma~15]{Nedic17}.] 
\end{lemma}
\begin{lemma}
\begin{align}
\|\tilde{x}\|_2^{a,K}\leq s\gamma_3\|\tilde{h}\|_2^{a,K}+\gamma_4
\label{x_to_h}
\end{align}
for some $\gamma_3$ and $\gamma_4$. [Precise values can be found in \cite[Lemma~16]{Nedic17}.] 
\end{lemma}
By using~\eqref{x_to_h},~\eqref{h_to_z}, and~\eqref{delta_to_z} in~\eqref{e_to_x}, we obtain
\begin{align}
\|e\|_2^{a,K}&\leq 2\|\tilde{x}\|_2^{a,K} \leq 2s\gamma_3\|\tilde{h}\|_2^{a,K}+2\gamma_4\nonumber\\&\leq 2s\gamma_1\gamma_3\|\delta\|_2^{a,K}+2s\gamma_2\gamma_3+2\gamma_4\nonumber\\&\leq 4s\gamma_1\gamma_3\hat{n}\|z\|_2^{a,K}+2s\gamma_2\gamma_3+2\gamma_4\nonumber\\&\quad+2s\gamma_1\gamma_3\|\delta[0]\|_2,
\end{align}
which can be rewritten as
\[\|e\|_2^{a,K}\leq s\alpha_2\|z\|_2^{a,K}+\beta_2,\]
where $\alpha_2 =  4\gamma_1\gamma_3\hat{n}$, and $\beta_2 = 2s\gamma_2\gamma_3+2\gamma_4+2s\gamma_1\gamma_3\|\delta[0]\|_2$,
yielding~\eqref{e_to_z}.
\end{proof}
In the following, we state the convergence results for algorithm~\eqref{primal_dual_alg}, which can be shown by applying the small-gain theorem to the results in Propositions~\ref{prop:H1_dd}--\ref{prop:H2_dd}, similar to the analysis in the proof of Proposition~\ref{prop:small_gain_un}.
\begin{proposition}\label{prop:small_gain_dd}
Let Assumptions~\ref{objective_assumption}, \ref{assume_comm_model1}, and \ref{assume_comm_model2} hold.
Then, under algorithm~\eqref{primal_dual_alg},
\begin{align*}
\|z\|_2^{a,K}\leq \beta,
\end{align*}
for some $\beta>0$, $a \in (0,1)$, sufficiently small~$s>0$, and $\forall\xi \in (0,\frac{n}{\hat{n}}]$. In particular, $(p_i[k],\lambda_i[k])$ converges to $(p_i^*,\lambda^*)$, $\forall i$, at a geometric rate~$\mathcal{O}(a^k)$.
\end{proposition}
Finally, in the following, we establish that $p^*$ is the solution of \eqref{ED}. [The proof is similar to the proof of Lemma~\ref{lem:p_star}]
\begin{lemma}\label{lem:p_star_dd}
Consider $(p^*,\lambda^*)$, namely, the equilibrium of the nominal system~$\vec{\mathcal{H}}_1$ with $e[k]\equiv 0$, $\forall k$.
Then, $p^*$ is the solution of \eqref{ED}.
\end{lemma}

\section{Robust DER Coordination Over Time-Varying Directed Graphs}\label{sec:robust-DOD-dd}
In this section, we present a robust extension of the distributed algorithm~\eqref{primal_dual_alg}, relaxing the assumption that each node knows its instantaneous out-degree. 
Instead of Assumption~\ref{assume_comm_model2}, we assume that each node only knows its nominal out-degree, as stated next.
\begin{assumption}
\label{assume_comm_model3}
The value of~$d_i^+$ is known to node~$i$, $i=1,2,\dots,n$, for all $k\geq 0$.
\end{assumption}

\subsection{Running-Sum Ratio Consensus Algorithms}
Since instantaneous out-degrees are not available, the ratio-consensus algorithm~\eqref{RatioCons} cannot be executed. If, instead of~$D_i^+[k]$, we use $d_i^+$ in~\eqref{RatioCons}, then, $P[k]$ is not necessarily column stochastic.
However, the loss of column-stochasticity can be fixed by augmenting the original network of nodes with additional virtual nodes and links such that if node~$i$ does not receive a packet from node~$j$, we let a virtual node receive the packet via a virtual link \cite{HaVaDo16}. This allows us to augment~\eqref{VectRatioCons} with additional states corresponding to the virtual nodes so that the augmented system becomes
\begin{align*}
\mu^\prime[k+1] &= \tilde{P}[k]\mu^\prime[k],\\
\nu^\prime[k+1] &= \tilde{P}[k]\nu^\prime[k],
\end{align*}
where $\mu^\prime[k]$ and $\nu^\prime[k]$ are the augmented state vectors that contain $\mu[k]$ and $\nu[k]$ and the states of the virtual nodes. The matrix~$\tilde{P}[k]$ can be made column stochastic by carefully updating the states of the virtual nodes. To explain this, we consider nodes~$i$ and $j$ connected via a communication link~$(j,i) \in \mathcal{E}^{(0)}$, and let $\mu^\prime_{ji}$ denote the state of the corresponding virtual node.
If $(j,i) \notin \mathcal{E}^{(c)}[k]$, then, $\mu^\prime_{ji}[k+1] = \mu^\prime_{ji}[k] + \frac{\mu_j[k]}{d_j^+}$, namely, the virtual node receives the packet from node~$j$.
If $(j,i) \in \mathcal{E}^{(c)}[k]$, then, we consider the following two options for updating the state of the virtual node.
\begin{enumerate}
\item[1.] The virtual node sends the value of its current state, $\mu_{ji}[k]$, to node~$i$, and sets the value of the next state to zero, i.e., $\mu_{ji}[k+1] = 0$. In the meantime, node~$j$ sends $\frac{\mu_j[k]}{d_j^+}$ to node~$i$.
\item[2.] The virtual node sends a portion of its current state, $\gamma\mu_{ij}[k]$, to node~$i$, where $\gamma$ is strictly positive and less than 1, and retains the other portion by performing the following update:
\[\mu^\prime_{ji}[k+1] = (1-\gamma)\mu^\prime_{ji}[k] + (1-\gamma)\frac{\mu_j[k]}{d_j^+},\]
where we notice that node~$j$ sends $(1-\gamma)\frac{\mu_j[k]}{d_j^+}$ to the virtual node and the remaining portion, $\gamma\frac{\mu_j[k]}{d_j^+}$, to node~$i$.
\end{enumerate}
The first option was chosen in the original running-sum ratio consensus algorithm (see, e.g., \cite{HaVaDo16}). 
In this work, we select the second option, since it allows us to considerably simplify the convergence analysis. 
To account for the absence of the virtual nodes and links in the actual communication network, additional computations must be performed at each transmitting/receiving node to effectively capture the effect of the updates at the virtual nodes on the states of the actual nodes.  
To this end, we let node~$j$ broadcast the running sums~$\sum_{t=0}^k {\mu_j[t]}/{d_j^+}$ and $\sum_{t=0}^k {\nu_j[t]}/{d_j^+}$. Then, $\mu_i[k]$ and $\nu_i[k]$ are updated by node~$i$ as follows:
\begin{subequations}\label{RunSumRatioCons}
\begin{align}
\mu_{i}[k+1] &= \sum_{j\in\mathcal{N}_i^-[k]\cup\{i\}}\big(\mu_{ij}[k+1]-\mu_{ij}[k]\big),\\
\nu_{i}[k+1] &= \sum_{j\in\mathcal{N}_i^-[k]\cup\{i\}}\big(\nu_{ij}[k+1]-\nu_{ij}[k]\big),\\
r_i[k+1] &= \frac{\mu_i[k+1]}{\nu_i[k+1]},
\end{align}
\end{subequations}
where $\mu_{ij}[k]$ and $\nu_{ij}[k]$ are updated using the running sums received by node~$i$ from node~$j$ and given by
\begin{subequations}\label{RunSumRatioCons2}
\begin{align}
\mu_{ij}[k+1] &= \left\{ \begin{matrix*}[l](1-\gamma)\mu_{ij}[k]+\gamma\sum\limits_{t=0}^k \frac{\mu_j[t]}{d_j^+}&\mbox{if }j\in\mathcal{N}_i^-[k],\\
\mu_{ij}[k] + \frac{\lambda_j[k]}{d_j^+}&\mbox{if }j = i,\\
\mu_{ij}[k]&\mbox{otherwise,}\end{matrix*}\right.\\
\nu_{ij}[k+1] &= \left\{ \begin{matrix*}[l](1-\gamma)\nu_{ij}[k]+\gamma\sum\limits_{t=0}^k \frac{\nu_j[t]}{d_j^+}&\mbox{if }j\in\mathcal{N}_i^-[k],\\
\nu_{ij}[k] + \frac{\nu_j[k]}{d_j^+}&\mbox{if }j = i,\\
\nu_{ij}[k]&\mbox{otherwise.}\end{matrix*}\right.
\end{align}
\end{subequations}
It is straightforward to see that the use of the running sums in the updates has the same effect on the states of the actual nodes as the updates at the virtual nodes have. 
Furthermore, by using the results in \cite{HaVaDo16}, it can be shown that $r_i[k]$ asymptotically converges with probability one to the average of the initial values, namely,
\begin{align*}
\lim_{k\rightarrow \infty}r_i[k] = \frac{\sum_{i}v_i}{n}.
\end{align*} 

\subsection{Robust Distributed Primal-Dual Algorithm}
By utilizing the running-sum ratio-consensus algorithm~\eqref{RunSumRatioCons}--\eqref{RunSumRatioCons2} in the averaging step, we develop a robust extension of algorithm~\eqref{primal_dual_alg}. We show that this robust extension is able to solve the DER coordination problem~\eqref{ED} even when every node~$i$ only knows its nominal out-degree, $d_i^+$, but not its instantaneous out-degree, $D_i^+[k]$.

We let node~$j$ broadcast the running sums~$\sum_{t=0}^k \frac{\lambda_j[t]}{d_j^+}$, $\sum_{t=0}^k \frac{v_j[t]}{d_j^+}$, and $\sum_{t=0}^k \frac{y_j[t]}{d_j^+}$ to its neighbors at each~$k\geq0$. Node~$i$ performs the following updates:
\begin{subequations}\label{robust_primal_dual_alg}
\begin{align}
p_i[k+1]&=\Big[p_i[k]-sf^\prime_i(p_i[k])+s\xi x_i[k]\Big]_{\underline{p}_i}^{\overline{p}_i},\\
\lambda_i[k+1] &= \sum_{j\in\mathcal{N}_i^-[k]\cup\{i\}}\Big(\lambda_{ij}[k+1]-\lambda_{ij}[k]-sy_{ij}[k+1]\nonumber\\&\quad+sy_{ij}[k]\Big),\\
v_i[k+1]&=\sum_{j\in\mathcal{N}_i^-[k]\cup\{i\}}\big(v_{ij}[k+1]-v_{ij}[k]\big),\\
x_i[k+1]&=\frac{\lambda_i[k+1]}{v_i[k+1]},\\
y_i[k+1]&=\sum_{j\in\mathcal{N}_i^-[k]\cup\{i\}}\big(y_{ij}[k+1]-y_{ij}[k]\big)\nonumber\\&\quad+\hat{n}(p_i[k+1]-p_i[k]).
\end{align}
\end{subequations}
where $\lambda_{ij}[k]$, $v_{ij}[k]$ and $y_{ij}[k]$ are updated using the running sums received by node~$i$ from node~$j$, and given by
\begin{subequations}\label{runsum}
\begin{align}
\lambda_{ij}[k+1] &= \left\{ \begin{matrix*}[l](1-\gamma)\lambda_{ij}[k]+\gamma\sum\limits_{t=0}^k \frac{\lambda_j[t]}{d_j^+} &\mbox{if }j\in\mathcal{N}_i^-[k],\\
\lambda_{ij}[k] + \frac{\lambda_j[k]}{d_j^+} &\mbox{if }j = i,\\
\lambda_{ij}[k] & \mbox{otherwise,}\end{matrix*}\right.\\
v_{ij}[k+1] &= \left\{ \begin{matrix*}[l](1-\gamma)v_{ij}[k] + \gamma\sum\limits_{t=0}^k \frac{v_j[t]}{d_j^+}& \mbox{if }j\in\mathcal{N}_i^-[k],\\
v_{ij}[k] + \frac{v_j[k]}{d_j^+}& \mbox{if }j = i,\\
v_{ij}[k] & \mbox{otherwise,}\end{matrix*}\right.\\
y_{ij}[k+1] &= \left\{ \begin{matrix*}[l](1-\gamma)y_{ij}[k] + \gamma\sum\limits_{t=0}^k \frac{y_j[t]}{d_j^+}& \mbox{if }j\in\mathcal{N}_i^-[k],\\
y_{ij}[k] + \frac{y_j[k]}{d_j^+}& \mbox{if }j = i,\\
y_{ij}[k] & \mbox{otherwise,}\end{matrix*}\right.
\end{align}
\end{subequations}
where $0<\gamma<1$.

\subsection{Feedback Representation of the Robust Distributed Primal-Dual Algorithm in Virtual Domain}
To facilitate the understanding of algorithm \eqref{robust_primal_dual_alg}--\eqref{runsum}, we represent it as a feedback interconnection of a nominal system, denoted by~$\vec{\mathcal{H}}_1^r$, and a disturbance system, denoted by~$\vec{\mathcal{H}}_2^r$. However, unlike the previously described feedback representations, this representation will be given in the virtual domain using the virtual nodes and links.

Consider a set of virtual nodes denoted by~$\mathcal{S} = \{n+1, \dots, n+|\mathcal{E}^{(0)}|\}$, where the virtual nodes correspond to the edges in $\mathcal{E}^{(0)}$ through a one-to-one map~$\mathds{I}$ such that $\mathds{I}(j,i) \in\mathcal{S}$ for $(j,i)\in\mathcal{E}^{(0)}$. 
Consider neighboring nodes~$i$ and $j$, i.e., $(j,i) \in \mathcal{E}^{(0)}$, and a virtual node~$l \in\mathcal{S}$ corresponding to the link from~$j$ to~$i$, i.e., $\mathds{I}(j,i) = l$. Let $\widetilde{\mathcal{N}}_l^-[k]$ denote the set of in-neighbors of node~$l$ at instant~$k$ given by~$\widetilde{\mathcal{N}}_l^-[k] = \{j\}$, $\forall k$, implying that node~$l$ always receives a packet from node~$j$.
Let $\widetilde{\mathcal{N}}_i^-[k]$ denote the augmented set of in-neighbors of node~$i\in\mathcal{V}$ at instant~$k$ given by
\begin{align}
\widetilde{\mathcal{N}}_i^-[k] = \mathcal{N}_i^-[k] \cup \{a \in\mathcal{S}: a=\mathds{I}(j,i), j \in\mathcal{N}_i^-[k] \}, \label{aug_nei_set}
\end{align}
which contains the set of in-neighbors~$\mathcal{N}_i^-[k]$ and the set of virtual nodes, from which node~$i$ receives a packet at instant~$k$. Note that the definition of~$\widetilde{\mathcal{N}}_i^-[k]$ in~\eqref{aug_nei_set} implies that node~$i$ receives a packet from node~$l$ at instant~$k$ if node~$i$ receives a packet from node~$j$ at instant~$k$.

If we let node~$l \in\mathcal{S}$ execute the following iterations:
\begin{subequations}\label{virtual_nodes}
\begin{align}
\lambda_l[k+1] &= \left\{ \begin{array}{l l}\lambda_l[k]+\frac{\lambda_j[k]}{d_j^+}& j\notin\widetilde{\mathcal{N}}_i^-[k],\\(1-\gamma)\lambda_l[k] + (1-\gamma)\frac{\lambda_j[k]}{d_j^+}& \mbox{otherwise,}\end{array}\right.\\
v_l[k+1] &= \left\{ \begin{array}{l l}v_l[k]+\frac{v_j[k]}{d_j^+}& j\notin\widetilde{\mathcal{N}}_i^-[k],\\(1-\gamma)v_l[k] + (1-\gamma)\frac{v_j[k]}{d_j^+}& \mbox{otherwise,}\end{array}\right.\label{virtual_v}\\
y_l[k+1] &= \left\{ \begin{array}{l l}y_l[k]+\frac{y_j[k]}{d_j^+}& j\notin\widetilde{\mathcal{N}}_i^-[k],\\(1-\gamma)y_l[k] + (1-\gamma)\frac{y_j[k]}{d_j^+}& \mbox{otherwise,}\end{array}\right.
\end{align}
where $\lambda_l[0]=0$, $v_l[0]=0$, and $y_l[0]=0$,
\end{subequations}
then, it is not difficult to see that the node~$i$'s updates in~\eqref{robust_primal_dual_alg} are equivalent to the following iterations:
\begin{subequations}\label{robust_primal_dual_alg3}
\begin{align}
p_i[k+1]&=\Big[p_i[k]-sf^\prime_i(p_i[k])+s\xi x_i[k]\Big]_{\underline{p}_i}^{\overline{p}_i},\\
\lambda_i[k+1] &= \frac{\lambda_i[k]-sy_i[k]}{d_i^+} + \sum_{a\in\widetilde{\mathcal{N}}_i^-[k] }\gamma\frac{\lambda_a[k]-sy_a[k]}{d_a^+},\\
v_i[k+1]&=\frac{v_{i}[k]}{d_i^+}+\sum_{a\in\widetilde{\mathcal{N}}_i^-[k]}\gamma\frac{v_{a}[k]}{d_a^+},\\
x_i[k+1]&=\frac{\lambda_i[k+1]}{v_i[k+1]},\\
y_i[k+1]&=\frac{y_{i}[k]}{d_i^+}+ \sum_{a\in\widetilde{\mathcal{N}}_i^-[k]}\gamma\frac{y_a[k]}{d_a^+}+\hat{n}(p_i[k+1]-p_i[k]),
\end{align}
\end{subequations}
where $d_a^+ \coloneqq 1$, $a \in \mathcal{S}$.
Now, we define $N\coloneqq n+|\mathcal{E}^{(0)}|$, and $\tilde{P}[k]\in\mathds{R}^{N\times N}$ such that
\begin{align*}
\tilde{P}_{ij}[k] &\coloneqq \left\{\begin{array}{l l} \frac{\gamma}{d_j^+} & \mbox{if }i\in\mathcal{V}, j\in\widetilde{\mathcal{N}}_i^-[k],\\
\frac{1-\gamma}{d_j^+} & \mbox{if }i\in\mathcal{S}, \mathds{I}(j,l)=i,j\in\widetilde{\mathcal{N}}_l^-[k],\\
\frac{1}{d_j^+} & \mbox{if }i\in\mathcal{S},\mathds{I}(j,l)=i, j\notin\widetilde{\mathcal{N}}_l^-[k],\\
0 & \mbox{else},\end{array}\right.\\
\tilde{P}_{ii}[k] &\coloneqq \left\{\begin{array}{l l}\frac{1}{d_i^+} & \mbox{if }i\in\mathcal{V},\\
\frac{1-\gamma}{d_i^+} & \mbox{if }i\in\mathcal{S},\mathds{I}(j,l)=i,j\in\widetilde{\mathcal{N}}_l^-[k],\\
\frac{1}{d_i^+} & \mbox{if }i\in\mathcal{S},\mathds{I}(j,l)=i,j\notin\widetilde{\mathcal{N}}_l^-[k].
\end{array}\right.
\end{align*}
Note that $\tilde{P}[k]$ is column stochastic.
Furthermore, for $i=1,\dots,N$, we have that
\begin{align}\label{P}
\tilde{P}_{ij}[k]&\geq \min(\gamma,1-\gamma)\min_{j\in\mathcal{V}\cup\mathcal{S}}\frac{1}{d_j^+}\nonumber\\&\geq \min(\gamma,1-\gamma)/n \coloneqq \tau, \mbox{ }j\in\widetilde{\mathcal{N}}_i^-[k]\cup\{i\}, \forall k,
\end{align}
where we used the fact that $d_j^+\leq n$, $\forall j$. This, in particular, implies that all diagonal entries in~$\tilde{P}[k]$ are always strictly positive.
For further development, we establish the following result using the analysis from the proof of \cite[Lemma~4]{Nedic15}. However, there are some subtle differences due to the fact that $v_i[0]=0$, for $i\in\mathcal{S}$. We recall that $v_i[0]=1$, for $i\in\mathcal{V}$.
\begin{lemma}\label{lem:v}
For $i=1,2,\dots,N$, we have that 
\begin{align}
v_i[k]\geq \frac{1-\gamma}{n}\tau^{N(2B-1)}, \mbox{ } \forall k\geq 1.\label{lem_v}
\end{align}
\end{lemma}
\begin{proof}
Since $\tilde{P}_{ii}[k] = 1/d_i^+$, $\forall i\in\mathcal{V}$, and $d_i^+\leq n$, we have that $\tilde{P}_{ii}[k]\geq 1/n$, $\forall i\in\mathcal{V}$ and $k\geq 0$. Hence,
\begin{align*}
(\tilde{P}[k+1]\dots \tilde{P}[0])_{ii}\geq \frac{1}{n}(\tilde{P}[k]\dots \tilde{P}[0])_{ii},
\end{align*}
for $i=1,\dots,n$.
Because $\tau<1/n$, it becomes clear that when $1\leq k\leq N(2B-1)$,  
\begin{align}\label{lem_v_ineq3}
(\tilde{P}[k-1]\dots \tilde{P}[0]\mathbf{v}[0])_i&\geq \tilde{P}_{ii}[k-1]\dots\tilde{P}_{ii}[0]\nonumber\\&\geq 1/n^{N(2B-1)}> \tau^{N(2B-1)}, 
\end{align}
for all $i\in\mathcal{V}$, where $\mathbf{v}[t] \coloneqq [v_1[t],v_2[t],\dots,v_{N}[t]]^\top$. 
We recall that, by \eqref{P}, we have that $\tilde{P}_{ij}[k]\geq \tau$, $i=1,\dots, N$, $\forall j\in\widetilde{\mathcal{N}}_i^-[k]\cup\{i\}$, $\forall k$.
Then, as shown in \cite[Lemma~2]{Nedic09}, for $t\geq(N-1)(2B-1)$, we have that 
\begin{align}\label{lem_v_ineq4}
(\tilde{P}[t-1]\dots \tilde{P}[0])_{ij}\geq\tau^{(N-1)(2B-1)}, \mbox{ }\forall i,j.
\end{align}
By combining \eqref{lem_v_ineq3} and \eqref{lem_v_ineq4} and using the fact that $v_i[0]=1$, for $i=1,\dots,n$, we find that
\begin{align}
v_i[k+1] = (\tilde{P}[k]\dots \tilde{P}[0]\mathbf{v}[0])_i\geq \tau^{N(2B-1)}, \mbox{ }\forall k\geq 0.\label{lem_v_ineq1}
\end{align}
Now, consider a virtual node~$l\in\mathcal{S}$ such that $\mathds{I}(i,j) = l$ for some $i,j \in\mathcal{V}$. Noticing that $i \in \widetilde{\mathcal{N}}_l^-[k]$, $\forall k\geq0$, and $d_i^+\leq n$, for $l=n+1,\dots,N$, we have from \eqref{virtual_v} that
\begin{align}
v_l[k+1]\geq \frac{1-\gamma}{d_i^+}v_i[k] \geq \frac{1-\gamma}{n}\tau^{N(2B-1)}, \mbox{ }\forall k\geq 0. \label{lem_v_ineq2}
\end{align}
Combining \eqref{lem_v_ineq1} and \eqref{lem_v_ineq2} yields \eqref{lem_v}.
\end{proof}
Next, we define additional virtual variables maintained by the virtual nodes. For $i\in\mathcal{S}$, we define
\begin{align*} 
x_i[k]&\coloneqq\left\{\begin{matrix*}[l] \frac{\lambda_i[k]}{v_i[k]}, & \mbox{if }k>0,\\ 0, & \mbox{if }k=0.\end{matrix*}\right.
\end{align*} 
We let $p_i[k]$ denote the iterate for the produced power at the virtual node~$i\in\mathcal{S}$ at instant~$k$, $f_i(p_i)\coloneqq mp_i^2$ denote the cost function, $\underline{p}_i=\overline{p}_i=0$ the capacity constraints, and $\ell_i = 0$ the consumed power. Since $\underline{p}_i=\overline{p}_i=0$, we have that $p_i[k]=0$, for all $k\geq0$. Since $p_i[k]=0$, $\forall k$, and $\ell_i = 0$, these virtual variables do not have any effect on the solution of the considered problem, and are only needed for describing the feedback system and allowing us to re-use the convergence results from Section~\ref{subsec:conv_analysis_dd}.

Next, we let
\begin{align*}
\mathbf{x}[k] &= [x_1[k],x_2[k],\dots,x_{N}[k]]^\top,\\
\bm{\lambda}[k] &= [\lambda_1[k],\lambda_2[k],\dots,\lambda_{N}[k]]^\top,\\
\mathbf{y}[k] &= [y_1[k],y_2[k],\dots,y_{N}[k]]^\top,\\
\mathbf{p}[k] &= [p_1[k],p_2[k],\dots,p_{N}[k]]^\top,\\
\bm{\ell} &= [\ell_1,\ell_2,\dots,\ell_{N}]^\top,\\
\mathbf{f}(\mathbf{p}[k]) &=  [f_1(p_1[k]),f_2(p_2[k]),\dots,f_{N}(p_{N}[k])]^\top,\\
\mathbf{\underline{p}} &= [\underline{p}_1,\dots,\underline{p}_{N}]^\top,\mathbf{\overline{p}} = [\overline{p}_1,\dots,\overline{p}_{N}]^\top,\\
v[k] &= [v_1[k],v_2[k],\dots,v_n[k]]^\top,\\ 
V[k] &\coloneqq \diag(v[k]),\tilde{V}[k] \coloneqq \diag(\mathbf{v}[k]).
\end{align*} 
Noticing that, by Lemma~\ref{lem:v}, $\tilde{V}[k]$ is invertible for all $k\geq 1$, we define
\begin{align} 
\mathbf{h}[k]\coloneqq \left\{\begin{matrix}(\tilde{V}[k])^{-1}\mathbf{y}[k] & \mbox{if }k>0,\\ \begin{bmatrix}(V[k])^{-1}y[k]\\\mathbf{0}_{N-n}\end{bmatrix}& \mbox{if }k=0,\end{matrix}\right.\label{def_h}
\end{align} 
where $y[k] = [y_1[k],y_2[k],\dots,y_n[k]]^\top$. [Instead of defining $\mathbf{h}[k]$ as $(\tilde{V}[k])^{-1}\mathbf{y}[k]$, we adopt the definition in \eqref{def_h}, because $\tilde{V}[k]$ is not invertible at $k=0$.]
Let $\tilde{R}[k]\coloneqq (\tilde{V}[k+1])^{-1}\tilde{P}[k]\tilde{V}[k]$, $I_n$ denote the $n\times n$ identity matrix, $\mathbf{0}_{a\times b}$ denote the $a\times b$ all-zeros matrix, and \[\tilde{I} = \begin{bmatrix}I_n & \mathbf{0}_{n\times (N-n)}\\ \mathbf{0}_{(N-n)\times n} & \mathbf{0}_{(N-n)\times (N-n)}\end{bmatrix}.\] 

Let $\mathbf{\overline{x}}[k] \coloneqq \mathbf{1}^\top\mathbf{x}[k]/N$, $\mathbf{\tilde{x}}[k]=\mathbf{x}[k]-\mathbf{1}\mathbf{\overline{x}}[k]$, and $\mathbf{\hat{x}}[k] \coloneqq \mathbf{1}^\top\mathbf{x}[k]/\hat{n}$. Then, by substituting $({\hat{n}}/{N})\mathbf{\hat{x}}[k]+\mathbf{\tilde{x}}[k]$ for $\mathbf{x}[k]$, $\tilde{V}[k]\mathbf{x}[k]$ for $\bm{\lambda}[k]$ and $(\tilde{V}[k+1])^{-1}\tilde{P}[k]\tilde{V}[k]$ for $\tilde{R}[k]$, and by using~\eqref{virtual_nodes}--\eqref{robust_primal_dual_alg3} and \eqref{def_h}, we obtain that
\begin{subequations}\label{robust_primal_dual_alg_vector}
\begin{align}
\mathbf{p}[k+1]&=\Big[\mathbf{p}[k]-s\nabla \mathbf{f}(\mathbf{p}[k])+s\xi\frac{\hat{n}}{N}\mathbf{\hat{x}}[k]\mathbf{1}+s\xi \mathbf{\tilde{x}}[k]\Big]_{\mathbf{\underline{p}}}^{\mathbf{\overline{p}}},\label{p_update_r}\\
\mathbf{x}[k+1]&=(\tilde{V}[k+1])^{-1}\big(\tilde{P}[k]\tilde{V}[k]\mathbf{x}[k]-s\tilde{I}\tilde{P}[k]\mathbf{y}[k]\big) \nonumber\\
&=(\tilde{V}[k+1])^{-1}\tilde{P}[k]\tilde{V}[k]\mathbf{x}[k]\nonumber\\&\quad-s\tilde{I}(\tilde{V}[k+1])^{-1}\tilde{P}[k]\mathbf{y}[k] \nonumber\\
&=(\tilde{V}[k+1])^{-1}\tilde{P}[k]\tilde{V}[k]\mathbf{x}[k]\nonumber\\&\quad-s\tilde{I}(\tilde{V}[k+1])^{-1}\tilde{P}[k]\tilde{V}[k]\mathbf{h}[k] \nonumber\\
&= \tilde{R}[k]\mathbf{x}[k]-s\tilde{I}\tilde{R}[k]\mathbf{h}[k],\label{x_update_r}\\
\mathbf{h}[k+1]&=(\tilde{V}[k+1])^{-1}\mathbf{y}[k+1]\nonumber\\
&=(\tilde{V}[k+1])^{-1}\big(\tilde{P}[k]\mathbf{y}[k]+\hat{n}(\mathbf{p}[k+1]-\mathbf{p}[k])\big)\nonumber\\
&=(\tilde{V}[k+1])^{-1}\big(\tilde{P}[k]\tilde{V}[k]\mathbf{h}[k]+\hat{n}(\mathbf{p}[k+1]-\mathbf{p}[k])\big)\nonumber\\&=\tilde{R}[k]\mathbf{h}[k]+\hat{n}(\tilde{V}[k+1])^{-1}(\mathbf{p}[k+1]-\mathbf{p}[k]).\label{h_update_r}
\end{align}
\end{subequations}
Let 
\begin{align}\label{h_def}
\mathbf{\overline{h}}[k] \coloneqq \frac{\mathbf{1}^\top\mathbf{y}[k]}{\mathbf{1}^\top\mathbf{v}[k]} = \frac{\hat{n}}{n}\mathbf{1}^\top(\mathbf{p}[k]-\ell).
\end{align}
Since $\tilde{R}[k]$ is row-stochastic \cite{Nedic17}, the following relations hold:
\begin{align}
\mathbf{\hat{x}}[k]&=\frac{1}{\hat{n}}\mathbf{1}^\top\mathbf{x}[k] =\frac{1}{\hat{n}}\mathbf{1}^\top\tilde{R}[k]\mathbf{1}\mathbf{\overline{x}}[k],\label{robust_eq1}\\
\mathbf{\overline{h}}[k]&=\frac{1}{n}\mathbf{1}^\top\tilde{I}\tilde{R}[k]\mathbf{1}\mathbf{\overline{h}}[k].
\label{robust_eq2}
\end{align}
By using \eqref{h_def}, \eqref{robust_eq1} and \eqref{robust_eq2}, we find from \eqref{x_update_r} that
\begin{align}
\mathbf{\hat{x}}[k+1]&=\frac{1}{\hat{n}}\mathbf{1}^\top\mathbf{x}[k+1]=\frac{1}{\hat{n}}\mathbf{1}^\top\tilde{R}[k]\mathbf{x}[k]-s\frac{1}{\hat{n}}\mathbf{1}^\top\tilde{I}\tilde{R}[k]\mathbf{h}[k]\nonumber\\
&= \mathbf{\hat{x}}[k] - s\mathbf{1}^\top(\mathbf{p}[k]-\ell)+ \Big(\frac{1}{\hat{n}}\mathbf{1}^\top\tilde{R}[k]\mathbf{x}[k]-\mathbf{\hat{x}}[k]\Big)\nonumber\\&\quad-s\Big(\frac{1}{\hat{n}}\mathbf{1}^\top\tilde{I}\tilde{R}[k]\mathbf{h}[k]-\frac{n}{\hat{n}}\mathbf{\overline{h}}[k]\Big)\nonumber\\
&= \mathbf{\hat{x}}[k] - s\mathbf{1}^\top(\mathbf{p}[k]-\ell) + \frac{1}{\hat{n}}\mathbf{1}^\top\tilde{R}[k]\big(\mathbf{x}[k]-\mathbf{1}\mathbf{\overline{x}}[k]\big)\nonumber\\&\quad-s\frac{1}{\hat{n}}\mathbf{1}^\top\tilde{I}\tilde{R}[k]\Big(\mathbf{h}[k]-\mathbf{1}\mathbf{\overline{h}}[k]\Big),\\
&= \mathbf{\hat{x}}[k] - s\mathbf{1}^\top(\mathbf{p}[k]-\ell) + \frac{1}{\hat{n}}\mathbf{1}^\top\tilde{R}[k]\mathbf{\tilde{x}}[k]\nonumber\\&\quad-s\frac{1}{\hat{n}}\mathbf{1}^\top\tilde{I}\tilde{R}[k]\mathbf{\tilde{h}}[k],
\label{ave_x_r}
\end{align}
where $\mathbf{\tilde{h}}[k]\coloneqq\mathbf{h}[k]-\mathbf{1}\mathbf{\overline{h}}[k]$.
Then, we use~\eqref{p_update_r} and~\eqref{ave_x_r} to determine the nominal system, $\vec{\mathcal{H}}_1^r$, as follows:
\begin{subequations}\label{H1_dd_r}
\begin{empheq}[left={\vec{\mathcal{H}}_1^r:\empheqlbrace}]{alignat=4}
    \mathbf{p}[k+1]&=\mbox{}&&\Big[\mathbf{p}[k]-s\nabla \mathbf{f}(\mathbf{p}[k])\nonumber\\&&&+s\xi\frac{\hat{n}}{N}\mathbf{\hat{x}}[k]\mathbf{1}+s\xi \mathbf{\tilde{x}}[k]\Big]_{\mathbf{\underline{p}}}^{\mathbf{\overline{p}}},\\
    \mathbf{\hat{x}}[k+1] &=\mbox{}&& \mathbf{\hat{x}}[k] - s\mathbf{1}^\top(\mathbf{p}[k]-\ell)\nonumber\\&&& + \frac{1}{\hat{n}}\mathbf{1}^\top\tilde{R}[k]\mathbf{\tilde{x}}[k]\nonumber\\&&&-s\frac{1}{\hat{n}}\mathbf{1}^\top\tilde{I}\tilde{R}[k]\mathbf{\tilde{h}}[k].
\end{empheq}
\end{subequations} 
Now, we use~\eqref{x_update_r} and~\eqref{h_update_r} to determine the disturbance system, $\vec{\mathcal{H}}_2$, as follows: 
\begin{subequations}\label{H2_dd_r}
\begin{empheq}[left={\vec{\mathcal{H}}_2^r:\empheqlbrace}]{alignat=4}
\mathbf{h}[k]&=\mbox{}&&\tilde{R}[k-1]\mathbf{h}[k-1]\nonumber\\&&&+\hat{n}(\tilde{V}[k])^{-1}(\mathbf{p}[k]-\mathbf{p}[k-1]),\label{H2_dd_r_1}\\
\mathbf{x}[k+1]&=\mbox{}&&\tilde{R}[k]\mathbf{x}[k]-s\tilde{I}\tilde{R}[k]\mathbf{h}[k],\label{H2_dd_r_2}\\
\mathbf{e}[k]&=\mbox{}&&[\mathbf{\tilde{x}}[k]^\top,\mathbf{\tilde{h}}[k]^\top]^\top.
\end{empheq}
\end{subequations} 
Then, as illustrated in Fig.~\ref{fig:control_interpret_dd_r}, algorithm~\eqref{robust_primal_dual_alg}--\eqref{runsum} can be viewed as a feedback interconnection of~$\vec{\mathcal{H}}_1^r$ and $\vec{\mathcal{H}}_2^r$, where $(\mathbf{p^*},\mathbf{x^*})$ is the equilibrium of~\eqref{H1_dd_r} when $\mathbf{e}[k]\equiv 0$, $\forall k$.
\begin{figure}
    \centering 
	\includegraphics[trim=0cm 0cm 0cm 0cm, clip=true, scale=0.6]{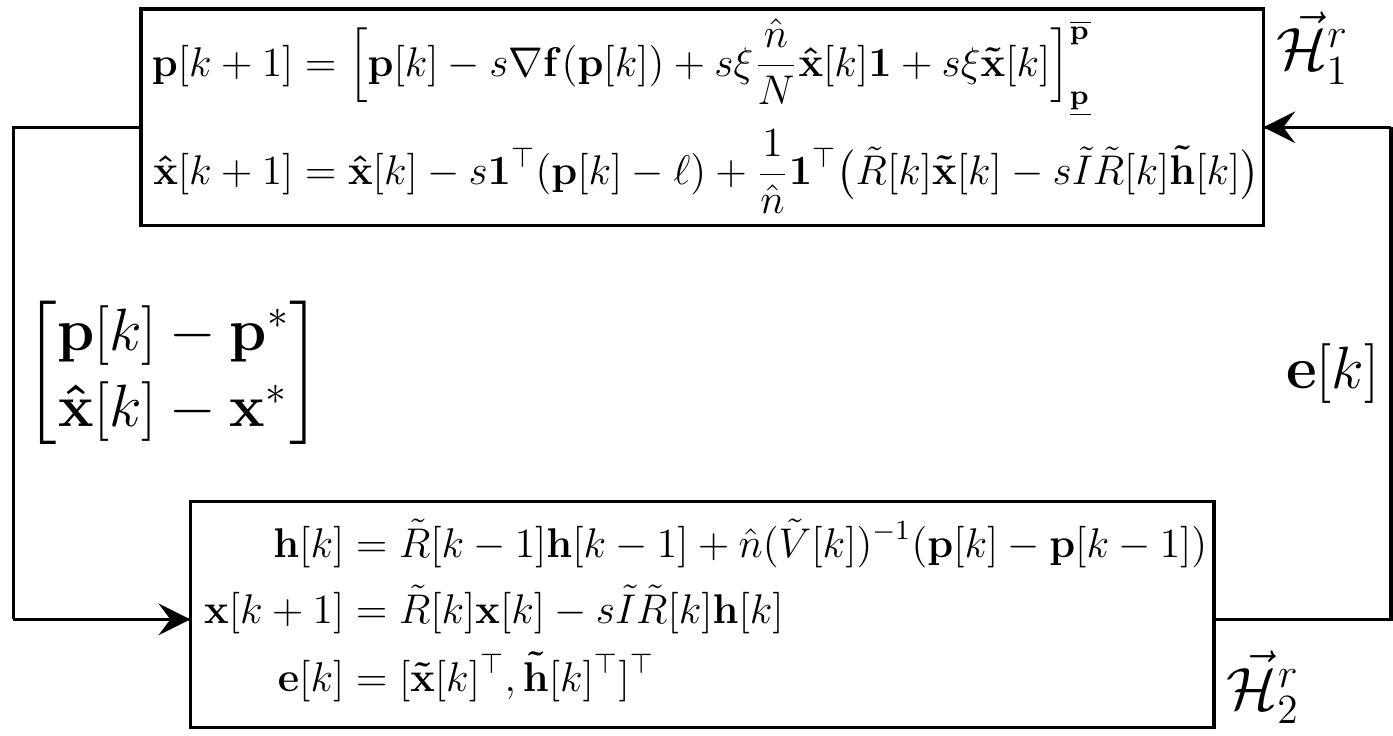} 
    \caption{Algorithm~\eqref{robust_primal_dual_alg}--\eqref{runsum} as a feedback system in the virtual domain.}
    \vspace{-5pt} 
    \label{fig:control_interpret_dd_r}
\end{figure}

\subsection{Convergence Analysis}
To establish the convergence results for algorithm~\eqref{robust_primal_dual_alg}--\eqref{runsum}, we show that $\vec{\mathcal{H}}_1^r$ and $\vec{\mathcal{H}}_2^r$ are finite-gain stable. This allows us to apply the small-gain theorem to prove that algorithm \eqref{robust_primal_dual_alg}--\eqref{runsum} converges to an optimal solution geometrically fast.

For our analysis, we need the following result, where we recall that $\tilde{R}[k]= (\tilde{V}[k+1])^{-1}\tilde{P}[k]\tilde{V}[k]$.
\begin{lemma}\label{lem:R}
For $i=1,\dots,N$, we have that
\begin{align}
\tilde{R}_{ij}[k]\geq \frac{1-\gamma}{n^2}\tau^{N(2B-1)+1}, \mbox{ }\forall j\in\widetilde{\mathcal{N}}_i^-[k]\cup\{i\}, k\geq 1.
\label{R_ineq}
\end{align}
\end{lemma}
\begin{proof}
From the definition of~$\tilde{R}[k]$, we have that $\tilde{R}_{ij}[k] = \tilde{P}_{ij}[k]v_j[k]/v_i[k+1]$. From Lemma~\ref{lem:v}, we have that \[v_j[k]\geq \frac{1-\gamma}{n}\tau^{N(2B-1)}, k\geq 1.\]
Since $\mathbf{1}^\top v[t]=n$, $\forall t\geq0$, it follows that $v_i[k+1]\leq n$.
We recall that, by \eqref{P}, $\tilde{P}_{ij}[k]\geq \tau$, $i=1,\dots, N$, $\forall j\in\widetilde{\mathcal{N}}_i^-[k]\cup\{i\}$, $\forall k$.
Hence,\[\tilde{R}_{ij}[k] \geq \frac{1-\gamma}{n^2}\tau^{N(2B-1)+1},\] $i=1,\dots, N$, $\forall j\in\widetilde{\mathcal{N}}_i^-[k]\cup\{i\}$, $\forall k$, yielding~\eqref{R_ineq}.
\end{proof}
By using \eqref{P} and the result in Lemma~\ref{lem:R}, the following lemmata can be established by borrowing much of the analysis from the proofs of \cite[Lemmas~15--16]{Nedic17}.
\begin{lemma}\label{lem:h_to_z_r}
\begin{align}
\|\mathbf{\tilde{h}}\|_2^{a,K} \leq \gamma_1\|\bm{\delta}\|_2^{a,K}+\gamma_2,
\label{h_to_z_r}
\end{align}
where $\bm{\delta}[k]\coloneqq \hat{n}(\mathbf{p}[k]-\mathbf{p}[k-1])$, for some~$\gamma_1$ and $\gamma_2$.
\end{lemma}
\begin{lemma}\label{lem:x_to_h_r}
\begin{align}
\|\mathbf{\tilde{x}}\|_2^{a,K}\leq s\gamma_3\|\mathbf{\tilde{h}}\|_2^{a,K}+\gamma_4
\label{x_to_h_r}
\end{align}
for some~$\gamma_3$ and $\gamma_4$.
\end{lemma}
Following the analysis from the proofs of Propositions~\ref{prop:H1_dd}--\ref{prop:small_gain_dd} and using Lemmas~\ref{lem:h_to_z_r}--\ref{lem:x_to_h_r}, we can easily establish the following results.
\begin{proposition}\label{prop:H1_dd_r}
Let Assumption~\ref{objective_assumption} hold.
Then, under~\eqref{H1_dd_r}, we have that
\begin{align}
\textnormal{\textbf{R1. }} \|\mathbf{z}\|_2^{a,K}\leq \alpha_1\|\mathbf{e}\|_2^{a,K}+\beta_1,
\label{R1_dd_r}
\end{align}
for some positive~$\alpha_1$ and $\beta_1$, $a \in (0,1)$, sufficiently small~$s>0$, and $\forall\xi \in (0,\frac{n}{\hat{n}}]$,
where \[\mathbf{z}[k] \coloneqq \begin{bmatrix}\mathbf{p}[k]-\mathbf{p}^*\\ \mathbf{\hat{x}}[k]-\mathbf{x}^* \end{bmatrix}.\]
\end{proposition}
\begin{proposition}\label{prop:H2_dd_r}
Let Assumptions~\ref{objective_assumption}, \ref{assume_comm_model1}, and \ref{assume_comm_model3} hold.
Then, under~\eqref{H2_dd_r}, we have that
\begin{align}
\textnormal{\textbf{R2. }}\|\mathbf{e}\|_2^{a,K}\leq s\alpha_2\|\mathbf{z}\|_2^{a,K}+\beta_2,
\label{R2_dd_r}
\end{align}
for some positive~$\alpha_2$ and $\beta_2$, $a \in (0,1)$, and sufficiently small~$s>0$.
\end{proposition}
\begin{proposition}\label{prop:small_gain_dd_r}
Let Assumptions~\ref{objective_assumption}, \ref{assume_comm_model1}, and \ref{assume_comm_model3} hold.
Then, under algorithm~\eqref{robust_primal_dual_alg}--\eqref{runsum},
\begin{align}
\|\mathbf{z}\|_2^{a,K}\leq \beta,
\label{z_relation_r}
\end{align}
for some~$\beta>0$, $a \in (0,1)$, sufficiently small~$s>0$, and $\forall\xi \in (0,\frac{n}{\hat{n}}]$. In particular, $(p_i[k],x_i[k])$ converges to $(\mathbf{p}_i^*,\mathbf{x^*})$, $i=1,\dots,n$, at a geometric rate~$\mathcal{O}(a^k)$.
\end{proposition}
Finally, in the following, we establish that $p^*\coloneqq[\mathbf{p}_1^*,\dots,\mathbf{p}_n^*]$ is the solution of \eqref{ED}. [The proof is similar to the proof of Lemma~\ref{lem:p_star}]
\begin{lemma}\label{lem:p_star_dd}
Consider $(\mathbf{p}^*,\mathbf{x}^*)$, namely, the equilibrium of the nominal system~$\vec{\mathcal{H}}^r_1$ with $\mathbf{e}[k]\equiv 0$, $\forall k$.
Then, $p^*$ is the solution of \eqref{ED}.
\end{lemma}

\subsection{Numerical Simulations}
\label{subsec:simulations_dd_r}
Next, we present the numerical results that illustrate the performance of the proposed robust distributed primal-dual algorithm~\eqref{robust_primal_dual_alg}--\eqref{runsum} using the same test system that was used previously in Section~\ref{subsec:simulations}. 
With regard to the communication model, every pair of nodes are connected by a single or two opposite unidirectional communication links if there is an electrical line between them. We assign the orientations of the communication links such that the nominal communication graph~$\mathcal{G}^{(0)}$ is strongly connected. Communication links fail with probability~$0.2$ independently (and independently between different time steps). We also assume that out-degrees,~$D_i^+[k]$, are unknown to DERs.

We compare the performance of algorithm~\eqref{robust_primal_dual_alg}--\eqref{runsum}, for convenience referred to as~$\mathbf{A}_1$, against that of the distributed algorithms proposed in \cite{DuYa18}, \cite{WuJo17}, \cite{KaHu12}, referred to as~$\mathbf{A}_2$, $\mathbf{A}_3$, and $\mathbf{A}_4$, respectively. In~$\mathbf{A}_1$, we use $\gamma=0.9$, $\hat{n}=20$, $s=0.02$, and $\xi=0.2$.

Algorithms~$\mathbf{A}_1$ and $\mathbf{A}_2$ use a constant stepsize~$s$. In contrast, $\mathbf{A}_3$ and $\mathbf{A}_4$ need to use a diminishing stepsize in order to guarantee convergence. However, if the stepsize is constant and sufficiently small, $\mathbf{A}_3$ and $\mathbf{A}_4$ can still achieve convergence within a small error. We tested the performance of~$\mathbf{A}_3$ using different diminishing stepsizes of the form~$s[k] = a/(k+b)$, where $a>0$ and $b>0$. To test $\mathbf{A}_4$, we used $(\alpha_k, \beta_k) = (0.003,0.3)$ (see, e.g., \cite{KaHu12}) and $(\alpha_k, \beta_k) = (\frac{20}{k+1000}, 0.3)$, which in this numerical example worked better than the stepsizes used in the numerical simulations in \cite{KaHu12}.   

In Fig.~\ref{fig:num_results_dd}, we provide the convergence error, namely, the Euclidean distance between the exact and iterative solutions, $\|p[k]-p^*\|_2$, for all algorithms. It can be seen that $\mathbf{A}_1$ outperforms $\mathbf{A}_2$, $\mathbf{A}_3$ and $\mathbf{A}_4$ and has geometric convergence speed. $\mathbf{A}_2$ fails to converge because out-degrees,~$D_i^+[k]$, are unknown to DERs. Through numerical simulations, we observed that it is in general difficult to choose the right values for~$a$ and $b$ in order for~$\mathbf{A}_3$ to operate well. In fact, if the ratio~$a/b$ is large, $\mathbf{A}_3$ might exhibit an oscillatory behavior. But setting~$a/b$ to a small value results in a slow convergence. 

\begin{figure}
    \centering
	\includegraphics[trim=0.2cm 0cm 0cm 0cm, clip=true, scale=0.48]{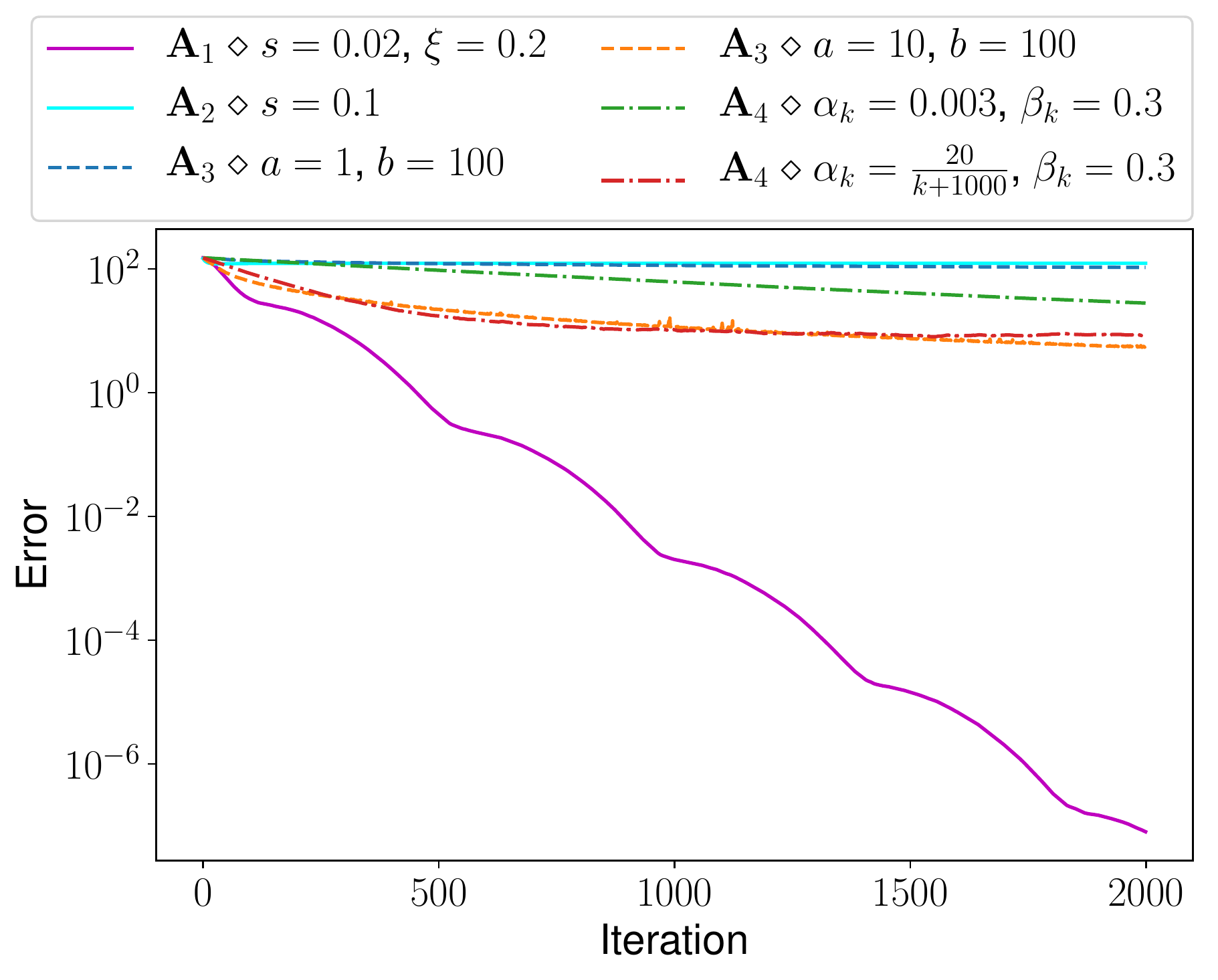} 
    \caption{Trajectory of~$\|p[k]-p^*\|_2$ for algorithms~$\mathbf{A}_1$--$\mathbf{A}_4$.}
    \vspace{-5pt} 
    \label{fig:num_results_dd}
\end{figure}

\section{Conclusion}
\label{sec:conclusion}
We presented distributed algorithms for solving the DER coordination problem over time-varying communication graphs. The algorithms have geometric convergence rate. One important future direction is to extend the proposed algorithms to solve more complex possibly multi-period DER coordination problems with additional constraints, e.g., line flow constraints, voltage constraints, or reactive power balance constraints.

\bibliographystyle{IEEEtran}
\bibliography{References}

\end{document}